\documentclass[11pt]{article}%
\pdfoutput=1
\usepackage{hyperref}
\usepackage{bbm,url,braket,microtype,amssymb,amsthm,mathtools,mathrsfs,cleveref,graphicx,float,authblk}
\usepackage{fullpage}
\usepackage{lmodern}
\usepackage[T1]{fontenc}
\usepackage[USenglish]{babel}
\usepackage[dvipsnames]{xcolor}
\usepackage[titletoc]{appendix}
\usepackage{tikz}
\usepackage{caption}
\usepackage{subcaption}
\usetikzlibrary{shapes.geometric}
\usetikzlibrary{angles}
\usetikzlibrary{quantikz, patterns}
\usetikzlibrary{calc}

\providecommand{\U}[1]{\protect\rule{.1in}{.1in}}
\newtheorem{thm}{Theorem}\crefname{thm}{Theorem}{Theorems}
\newtheorem{lem}[thm]{Lemma}\crefname{lem}{Lemma}{Lemmas}
\newtheorem{prp}[thm]{Proposition}\crefname{prp}{Proposition}{Propositions}
\newtheorem{cor}[thm]{Corollary}\crefname{cor}{Corollary}{Corollaries}
\crefname{prb}{Problem}{Problems}
\newtheorem{dfn}[thm]{Definition}\crefname{dfn}{Definition}{Definitions}
\newtheorem{rmk}[thm]{Remark}\crefname{rmk}{Remark}{Remarks}
\crefname{section}{Section}{Sections}
\crefname{appendix}{Appendix}{Appendices}
\numberwithin{equation}{section}
\let\oldref\ref
\renewcommand{\ref}[1]{(\oldref{#1})}
\DeclareMathOperator{\tr}{tr}
\DeclareMathOperator{\spec}{spec}

\newcommand{\abs}[1]{{\left\vert{#1}\right\vert}}

\title{Randomized Benchmarking Beyond Groups}

\author[1]{Jianxin Chen}
\author[2]{Dawei Ding}
\author[1]{Cupjin Huang}
\affil[1]{Alibaba Quantum Laboratory, Alibaba Group USA, Bellevue, Washington 98004, USA}
\affil[2]{Alibaba Quantum Laboratory, Alibaba Group USA, Sunnyvale, California 94085, USA}

\date{}

\begin{document}

\maketitle
\begin{abstract}
Randomized benchmarking (RB) is the gold standard for experimentally evaluating the quality of quantum operations. The current framework for RB is centered on groups and their representations, but this can be problematic. For example, Clifford circuits need up to $O(n^2)$ gates, and thus Clifford RB cannot scale to larger devices. Attempts to remedy this include new schemes such as linear cross-entropy benchmarking (XEB), cycle benchmarking, and non-uniform RB, but they do not fall within the group-based RB framework. In this work, we formulate the \emph{universal randomized benchmarking (URB) framework} which does away with the group structure and also replaces the recovery gate plus measurement component with a general ``post-processing'' POVM. Not only does this framework cover most of the existing benchmarking schemes, but it also gives the language for and helps inspire the formulation of new schemes. We specifically consider a class of URB schemes called \emph{twirling schemes}. For twirling schemes, the post-processing POVM approximately factorizes into an intermediate channel, inverting maps, and a final measurement. This leads us to study the twirling map corresponding to the gate ensemble specified by the scheme. We prove that if this twirling map is strictly within unit distance of the Haar twirling map in induced diamond norm, the probability of measurement as a function of gate length is a single exponential decay up to small error terms. The core technical tool we use is the matrix perturbation theory of linear operators on quantum channels.
\end{abstract}

\tableofcontents
\section{Introduction}
With recent breakthroughs in the hardware development of quantum processors, it becomes increasingly important to be able to efficiently characterize their performance. Such a task, often called \emph{benchmarking}, is essential in directing future quantum hardware development. An efficient and reliable benchmarking scheme not only allows comparison between different physical platforms and designs, but also provides useful feedback for device calibration and error diagnosis. This in turn provides useful information for future hardware designs, and, eventually, achieving fault-tolerant quantum computing.

A large number of such benchmarking schemes is collectively called \emph{randomized benchmarking} (RB). Randomized benchmarking aims to extract information about a certain gate or a collection of gates through runs of random gate sequences, while isolating the effect of state preparation and measurement (SPAM) errors. Implementing the RB experiment usually produces an exponential decay curve with respect to the length of the random gate sequences. It is widely assumed that the decay rate of this curve, which can be obtained via fitting the experimental data, indicates a certain fidelity measure of the given gate set~\cite{emerson2005scalable}. Since its first proposal in~\cite{emerson2005scalable}, randomized benchmarking has been thoroughly analyzed for its theoretical correctness and robustness~\cite{knill2008randomized, magesan2011scalable, wallman2018randomized, merkel2021randomized, helsen2020general}, and it has been extensively validated by experiments~ \cite{knill2008randomized,gaebler2012randomized,gambetta2012characterization,mckay2019three,garion2021experimental}. Many variants of the original RB protocol have been proposed due to its great flexibility and adaptability to different experimental scenarios~\cite{magesan2011scalable,carignan2015characterizing, cross2016scalable, wallman2015estimating, wallman2016robust}, 

One central question in the study of RB is choosing the probability distributions according to which the random gate sequences are chosen. Most RB schemes require that the gate set forms a group,\footnote{In fact, most RB schemes specify that all but the final gates are drawn i.i.d.\ uniformly or Haar-randomly from the group.} as the final gate in most RB experiments inverts all the previously applied gates, hence requiring the gate set to be closed under inversion of products.  However, requiring that the gate set form a group is not always feasible nor necessary in an experimental setting. We give several reasons:
\begin{itemize}
    \item Realizing an inverse gate can be experimentally challenging. In superconducting devices, arbitrary single-qubit gates can be easily implemented with high fidelities, whereas realizing a new two-qubit gate usually requires significant additional control and calibration. As a result, a connected qubit pair usually only admits a single calibrated two-qubit gate. For two-qubit gates that are locally equivalent to their inverses (or equivalently, locally equivalent to a real rotation in $SO(4)$~\cite{zhang2003geometric}), inverting them only requires appending the corresponding single-qubit gates. However, this is not true for all two-qubit gates. Counterexamples include the fSim gates which are indeed used in benchmarking experiments~\cite{arute2019quantum}.
    \item Group structure can be very costly to implement. A typical Haar random element in $SU(2^n)$ would require a circuit of exponential length and depth to realize. A random Clifford gate on $n$ qubits on average requires $\Omega(n^2)$ size and $\Omega(n/\log n)$ depth~\cite{jiang2020optimal}. Neither of these groups can be used to extract any useful information given the noise levels on current physical devices; the signal would be long gone after implementing just a few group elements.
    \item Recently, linear cross-entropy benchmarking~\cite{arute2019quantum} was proposed as an experiment-friendly alternative to RB on large quantum devices with tens of qubits. Linear XEB uses random shallow circuits whose unitaries clearly do not form a group, but an exponential decay was nevertheless observed experimentally. In fact, multiple existing benchmarking schemes~\cite{proctor2021scalable, proctor2019direct, erhard2019characterizing} rely on random sequences of shallow circuits whose unitaries do not form groups.
\end{itemize}

In this paper, we extend the RB framework beyond groups, and propose a \emph{universal randomized benchmarking (URB) framework} which incorporates most of existing RB schemes,\footnote{With the exception of leakage benchmarking protocols, the correct formulation of which might require quantum channels on infinite-dimensional Hilbert spaces. } in particular ones that cannot be easily formulated as group-based protocols. Most importantly, the random gates do not need to form a group but can be any set. Furthermore, instead of a single recovery gate that is often applied at the end of the group-based RB sequence (which is often well-defined given the group structure), we allow for a more general \emph{post-processing POVM} that depends on the random set elements chosen. For known RB schemes, intuitively the post-processing POVM \emph{verifies} that the gate sequence was perfectly applied, which is the case for group-based  RB where the post-processing POVM factors into the recovery gate and the final measurement. However, under our framework any gate sequence dependent POVM is allowed, even ones that do not follow the intuition of verification. The general, possibly not a group, gate set and the general post-processing POVM define the essence of the generality provided by the URB framework. This helps open the doors to thinking about RB from a new angle and may motivate fundamentally different schemes.

Our main technical contribution is determining conditions under which a URB scheme experiment gives rise to a single exponential decay. The usual strategy to prove exponential decay for an RB scheme is to first formulate the RB measurement probability as a linear functional of powers of a linear operator. The linear operator represents the average effect of one random gate, and the power is the sequence length. The spectral properties of the linear operator then gives the exponential decay. For group-based RB with gate-independent noise, this operator is the twirled noise channel~\cite{emerson2005scalable,magesan2011scalable, carignan2015characterizing}. For group-based RB with gate-dependent noise it is the Fourier transformation of the implementation map~\cite{merkel2021randomized,helsen2020general,kong2021framework}. For our case where a group structure is not present, we assume the post-processing POVM involves in a certain sense inverting the random gates, leading us to study linear operators on the set of quantum channels known as \emph{twirling maps}:
$$\Lambda^*_R:\mathcal{N}\mapsto \int_{g\sim \mu}dg \omega(g)^\dagger\circ \mathcal{N}\circ \omega(g),$$
where $\mu$ is some measure over a gate set and $\omega$ is a map from the gate set to $SU(d)$. Intuitively, these are higher-order operators representing the averaged joint action of a random gate and its inversion. They appeared in~\cite{boone2019randomized} to analyze a specific scheme without group structure. We can prove that single exponential decays can be observed in general when the corresponding twirling map is within $\gamma$ of the Haar twirling map (set $\mu$ as the Haar measure over $SU(d)$) in induced diamond norm with $\gamma<1$. This does not involve any underlying group structure, but instead the matrix perturbation theory of higher-order operators.


The rest of the paper is organized as follows. We first introduce basic notation, introducing operator norms and matrix perturbation theory of higher-order operators in \Cref{sec:prelim}. We then introduce the URB framework, how an experiment would run, and state sufficient conditions to obtain a single exponential decay in \Cref{sec:framework}, together with examples of existing benchmarking schemes formulated as URB schemes. \Cref{sec:result} presents the proof of the single exponential decay under the assumptions and other relevant technical details. \Cref{sec:discussion} concludes with a discussion and open problems.

\section{Preliminaries}
\label{sec:prelim}
\subsection{Linear Operators and Norms}
Here we give an exposition of different norms on linear operators that we will use.

\paragraph{Norms on Hermitian matrices.}
The spectrum of a Hermitian matrix is a vector in $\mathbb{R}^d$. We can then define norms of a Hermitian matrix via vector $\ell_p$-norms of its spectrum. For general matrices, these are known as Schatten $p$-norms. We consider three norms:
\begin{itemize}
    \item Trace norm $$\Vert \rho \Vert_{1}:= \Vert \spec(\rho)\Vert_{\ell_1} =\tr \vert \rho \vert.$$
    \item Frobenius norm $$\Vert \rho \Vert_{F}:= \Vert \spec(\rho)\Vert_{\ell_2} =\sqrt{\tr [\rho^2]}.$$
    \item Spectral norm $$\Vert \rho \Vert_{\infty}:= \Vert \spec(\rho)\Vert_{\ell_\infty} = \max_i \vert \lambda_i(\rho)\vert.$$
\end{itemize}
It is easy to see that
\begin{align}
\label{eq:hermitian_norm_inequality}
    \Vert\rho\Vert_\infty\leq\Vert\rho\Vert_F\leq\Vert \rho\Vert_{1}\leq \sqrt{d}\cdot \Vert\rho\Vert_F.
\end{align}
Quantum states are positive semidefinite operators with unit trace, and a positive operator-valued measurement (POVM) element is a positive semidefinite operator with spectral norm upper bounded by 1.

Define the Hilbert-Schmidt inner product on Hermitian matrices as $\langle\cdot,\cdot\rangle_{HS}:(\rho,\sigma)\mapsto\tr[\rho\sigma]$. It is easy to see that the Frobenius norm of Hermitian matrices is induced by this inner product.

\paragraph{Norms on real superoperators}
Consider linear maps $\mathcal{C}$ on Hermitian matrices, which we call \emph{real superoperators}. We can treat them as usual linear maps and define the corresponding operator norms:
\begin{itemize}
    \item Induced trace norm $$\Vert\mathcal{C}\Vert_{\tr} :=\max_{\rho \neq 0}\frac{\Vert \mathcal{C}(\rho)\Vert_{1}}{\Vert\rho\Vert_{1}}.$$
    \item Induced Frobenius norm $$\Vert\mathcal{C}\Vert_{2} :=\max_{\rho \neq 0}\frac{\Vert \mathcal{C}(\rho)\Vert_{F}}{\Vert\rho\Vert_{F}}.$$
\end{itemize}
We can also define another norm via the \emph{superoperator} inner product: $$\langle\cdot,\cdot\rangle_\mathrm{SO}:(\mathcal{C},\mathcal{D})\mapsto\sum_i\langle \mathcal{C}(X_i),\mathcal{D}(X_i)\rangle_\mathrm{HS},$$
where $\{X_i\}_i$ is an orthonormal basis of matrices under the Hilbert-Schimidt inner product. Note that $\{X_i\}_i$ can be any orthonormal basis of matrices, and we can extend the action of $\mathcal C$ to non-Hermitian matrices via linearity. We can then define
\begin{align*}
    \Vert \mathcal C \Vert_\mathrm{SO} := \sqrt{\langle \mathcal C,\mathcal C\rangle_\mathrm{SO}}.
\end{align*}
This norm is the analogue of the Frobenius norm for matrices. Another norm, called the \emph{diamond norm}, is defined on composite systems:
$$\Vert\mathcal{C}\Vert_{\diamond} :=\max_{\rho \neq 0}\frac{\Vert (\mathcal{C}\otimes\mathrm{id})(\rho)\Vert_{1}}{\Vert\rho\Vert_{1}},$$
where $\mathrm{id}$ is the identity map on $\mathbb{C}^{d\times d}$.

From the definitions and~\Cref{eq:hermitian_norm_inequality} we have the following relations between the norms:
\begin{align}
    \Vert\mathcal{C}\Vert_{\tr}&\leq\Vert \mathcal{C}\Vert_{\diamond},\nonumber \\
\sqrt{d^{-1}}\Vert \mathcal{C}\Vert_{2}&\leq \Vert \mathcal{C}\Vert_{\tr}\leq \sqrt{d}\Vert \mathcal{C}\Vert_{2},\nonumber\\
 d^{-1} \Vert \mathcal C\Vert_2 &\leq \Vert \mathcal{C}\Vert_{\diamond}\leq d\Vert \mathcal{C}\Vert_{2}.\label{eqn:channel_dim}
\end{align}
We can also prove the following norm inequality:
\begin{align*}
    & \Vert \mathcal C\Vert_\mathrm{SO}^2 \nonumber\\
    & =  \sum_{i,j} \Vert\mathcal C(E(i,j))\Vert_F^2 \nonumber\\
    & \leq \sum_{i,j} \Vert \mathcal C(E(i,j))\Vert_1^2 \nonumber\\
    & = \sum_{i,j} \Vert \mathcal C(E(i,j)_H + E(i,j)_{AH})\Vert_1^2 \nonumber\\
    &  \leq \sum_{i,j} \Vert \mathcal C(E(i,j)_H) \Vert^2_1+ 2 \Vert \mathcal C(E(i,j)_H) \Vert_1 \Vert i \mathcal C(E(i,j)_{AH} /i )\Vert_1+ \Vert i \mathcal C(E(i,j)_{AH} /i)\Vert_1^2 \nonumber\\
    & \leq 4 d^2 \Vert \mathcal C\Vert_{\tr}^2,
\end{align*}
where $E(i,j)_{kl} := \delta_{ik}\delta_{jl}$ are the elementary matrices, which are orthonormal under the Hilbert-Schmidt inner product and have unit trace norms, and $_H$ and $_{AH}$ denote the Hermitian and anti-Hermitian parts, respectively. We conclude
\begin{align}
    \label{eq:SO_tr_inequality}
    \Vert \mathcal C \Vert_\mathrm{SO} \leq 2 d \Vert \mathcal C \Vert_{\tr}.
\end{align}

For the other direction, we argue
\begin{align*}
\Vert \mathcal{C}\Vert_{\tr} \leq \sqrt{d} \Vert \mathcal C \Vert_2 \leq \sqrt{d}\Vert \mathcal{C}\Vert_{SO},
\end{align*}
where the second inequality follows since any Hermitian matrix can be extended to a basis.

A quantum channel is a completely positive, trace-preserving (CPTP) map, and consequently has unit induced trace norm and unit diamond norm. Furthermore, a quantum channel has unit induced Frobenius norm if it is a unitary, and sub-unit induced Frobenius norm if it is a mixture of unitaries. Throughout the paper we will mainly consider the Hilbert space spanned by all channels equipped with the $\langle \cdot,\cdot\rangle_\mathrm{SO}$ inner product, which we denote as $V(d)$. There is a subspace $V_0(d)$ with codimension 1 spanned by all \emph{differences} of quantum channels. This with an arbitrary quantum channel spans the whole $V(d)$. All norms defined on real superoperators natrually carries to $V(d)$ and $V_0(d)$.

It can be verified that the above norms are all bona fide matrix norms, namely
$$\Vert\mathcal{C}+\mathcal{D}\Vert\leq\Vert\mathcal{C}\Vert+\Vert\mathcal{D}\Vert, \Vert\mathcal{C}\circ\mathcal{D}\Vert\leq \Vert\mathcal{C}\Vert\Vert\mathcal{D}\Vert$$
for arbitrary $\mathcal{C}$ and $\mathcal{D}$. Furthermore, for the SO norm and the induced Frobenius norm we have
\begin{align}
\label{eq:SO_2_submultiplicative}    
\Vert \mathcal{C}\circ\mathcal{D}\Vert_{SO}\leq \Vert \mathcal{C}\Vert_2\Vert \mathcal{D}\Vert_\mathrm{SO}, \Vert \mathcal C\Vert_\mathrm{SO} \Vert \mathcal D\Vert_2.
\end{align}

\paragraph{Norms on linear maps on real superoperators}
In this work, we investigate \emph{linear operators on} real superoperators, which we call \emph{twirling maps}. Again, we can treat them like linear operators. One can define the  \emph{induced diamond norm} of a twirling map $\Lambda$ as\footnote{Without specifying otherwise, all norms are induced from real superoperators. Note that restricting to subspaces $V(d)$ or $V_0(d)$ does not change the inequalities and does not increase the induced norms.}
$$||| \Lambda|||_{\diamond}:=\max_{\mathcal{C}\neq 0}\frac{\Vert \Lambda(\mathcal{C})\Vert_{\diamond}}{\Vert \mathcal{C}\Vert_{\diamond}}.$$
We can similarly define the induced SO norm $|||\cdot|||_2:=\Vert\cdot\Vert_{SO\rightarrow SO}$ and induced trace norm $|||\cdot|||_{\tr}:=\Vert\cdot\Vert_{\tr{}\rightarrow \tr{}}$, where to avoid awkwardness, we omitted the second ``induced''. Again, these are all bona fide matrix norms. Note that the $||| \cdot |||_2$ norm corresponds to the usual spectral norm under the superoperator inner product. Although we do not use them here, a detailed discussion involving norms and approximate twirls and unitary designs is given in~\cite{low2010pseudo}. 
We have the following relation similar to \Cref{eqn:channel_dim}:
\begin{align}
d^{-3/2}||| \Lambda|||_{2}\leq& ||| \Lambda|||_{\tr}\leq d^{3/2}||| \Lambda|||_{2}.\label{eqn:twirl_dim}
\end{align}
Moreover, twirling maps have the following property.
\begin{prp}[Data Processing Inequality]\label{prp:diamond_norm_bound}
If a twirling map $\Lambda$ can be decomposed as
$$\Lambda(\cdot)=\int_g dg \mathcal{C}(g)\circ\cdot\circ \mathcal{D}(g),$$
then 
$$||| \Lambda|||\leq \int_g dg\Vert \mathcal{C}(g)\Vert\Vert \mathcal{D}(g)\Vert,$$
where we have the appropriate correspondence between twirling map and real superoperator norms. Moreover, we have a tighter bound regarding the induced SO norm on twirling maps and induced Frobenius norm on real superoperators:
$$|||\Lambda|||_2\leq \int_{g}dg\Vert \mathcal{C}(g)\Vert_2\Vert \mathcal{D}(g)\Vert_2.$$
\end{prp}
\begin{proof}
    The first inequality follows from the triangle inequality plus submultiplicativity, while the second inequality follows from the triangle inequality and~\Cref{eq:SO_2_submultiplicative}.
\end{proof}

\subsection{Matrix Perturbation Theory}
We first state a few results regarding the perturbation of matrices. Let $\mathcal{H}$ be a finite dimensional Hilbert space and $V_1, V_2$ be subspaces such that $V_1 \oplus V_2 = \mathcal H$. Let $A_1$ and $A_2$ be operators on $V_1$ and $V_2$ respectively. Let $\mathcal{M}(V_2,V_1)$ be the set of linear operators from $V_2$ to $V_1$, that is, linear operators $P=X_1PX_2$ where $X_1$ and $X_2$ are the projectors onto $V_1$ and $V_2$ respectively. For a norm $\Vert\cdot\Vert$ defined on $\mathcal{H}$, we define the corresponding seperation function $\mathrm{sep}(\cdot,\cdot)$ as
$$\mathrm{sep}(A_1,A_2):=\inf_{P\in \mathcal{M}(V_2,V_1)}\setminus\{0\}\frac{\Vert A_1P-PA_2\Vert}{\Vert P\Vert}.$$ From~\cite{stewart1990matrix} we have the following result.
\begin{thm}[Stewart and Sun~\cite{stewart1990matrix}]\label{thm:perturbation}
Let $\mathcal{H}$ be a finite-dimensional Hilbert space with norm $\Vert \cdot\Vert$. This norm naturally induces a norm on linear operators. Let $V_1$ be a linear subspace of $\mathcal{H}$ and $V_2$ be its orthogonal complement.  Let $X_1$ and $X_2$ be the projectors onto $V_1$ and $V_2$ respectively. Let $A$ be a linear operator on $\mathcal{H}$ such that 
$$X_i AX_j=A_i\delta_{ij}, i,j\in\{1,2\}.$$
Let $E$ be an arbitrary operator. If $E$ satisfies 
$$\mathrm{sep}(A_1,A_2)-\Vert X_1EX_1\Vert - \Vert X_2 EX_2\Vert>0,$$
$$\frac{\Vert X_1EX_2\Vert\Vert X_2EX_1\Vert}{(\mathrm{sep}(A_1,A_2)-\Vert X_1 EX_1\Vert-\Vert X_2 EX_2\Vert)^2}<\frac{1}{4},$$
then there exist operators $P_1, P_2$ such that 
$$X_2 P_1X_1=P_1, X_1P_2X_2=P_2,$$
$$\Vert P_1\Vert\leq \frac{2\Vert X_2EX_1\Vert}{\mathrm{sep}(A_1,A_2)-\Vert X_1EX_1\Vert-\Vert X_2EX_2\Vert},$$
$$\Vert P_2\Vert\leq \frac{\Vert X_2 EX_1\Vert}{\mathrm{sep}(A_1,A_2)-\Vert X_1EX_1\Vert-\Vert X_2EX_2\Vert-2\Vert P_1\Vert\Vert X_1EX_2\Vert},$$
such that $A+E$ can be diagonalized as
$$L_i^\dagger (A+E) R_j = A'_i \delta_{ij},$$
where $$L_i^\dagger = \begin{pmatrix}
I&-P_2\\0&I
\end{pmatrix}\begin{pmatrix}
I&0\\-P_1&I
\end{pmatrix}X_i^\dagger,$$
$$R_i = X_i\begin{pmatrix}
I&0\\P_1&I
\end{pmatrix}\begin{pmatrix}
I&P_2\\0&I
\end{pmatrix},i\in\{1,2\},$$
and
$$A'_1=A_1+X_1EX_1+X_1EP_1,$$
$$A'_2=A_2+X_2EX_2+P_1EX_2.$$
\end{thm}
\noindent We will actually make use of a corollary that has stronger assumptions:
\begin{cor}
\label{cor:simplified_perturbation}
Let $\mathcal{H}$ be a finite-dimensional Hilbert space with norm $\Vert \cdot\Vert$. This norm naturally induces a norm on linear operators.  Let $X_1, A, X_2 = I-X_1$ be linear operators on $\mathcal{H}$ such that\\
\noindent (1) $X_1$  is an orthogonal projector, (2) $X_i AX_j=A_i\delta_{ij}, i,j\in\{1,2\}$, (3) $\Vert A_2\Vert = \gamma<1$,\\ (4) $\Vert E\Vert\leq \delta$, (5) $A_1 = X_1$, (6) $\Vert X_1\Vert \leq 1$, (7) $\delta \leq  \frac{1-\gamma}{11}$. 

\noindent Then we obtain the conclusion of~\Cref{thm:perturbation} with $\Vert P_1\Vert \leq 1,\Vert P_2\Vert\leq 1$. Furthermore,
\begin{itemize}
    \item All eigenvalues of $A'_1$ is $2\delta$-close to $1$,
    \item $\Vert A'_2\Vert\leq  (\gamma+6\delta)$,
    \item 
 $\kappa :=\Vert L_2\Vert\Vert R_2\Vert\leq 16$.
\end{itemize}
\end{cor}

\begin{proof}
    We first prove $\mathrm{sep}(A_1,A_2)\geq1-\gamma$. This is because $A_1P=X_1P=P$ and thus
    $$\mathrm{sep}(A_1,A_2):=\inf_{P\in \mathcal{M}(V_2,V_1),\Vert P\Vert = 1}1-\Vert PA_2\Vert\geq 1-\gamma.$$
    Furthermore, $\Vert X_2\Vert = \Vert I-X_1\Vert \leq 2$. Therefore $\Vert X_1EX_1\Vert\leq \delta, \Vert X_2 EX_1\Vert \leq 2\delta, \Vert X_1EX_2\Vert\leq 2\delta$ and $\Vert X_2 EX_2\Vert\leq 4\delta$. It is easy to verify that both assumptions in \Cref{thm:perturbation} hold, and 
    $$\Vert P_1\Vert\leq \frac{4\delta}{11\delta-5\delta}\leq 1,$$
    $$\Vert P_2\Vert\leq \frac{2\delta}{11\delta-5\delta-4\delta}\leq 1.$$
    We have
    $$\Vert A'_1-A_1\Vert\leq \Vert X_1\Vert \Vert E\Vert (\Vert X_1\Vert + \Vert P_1\Vert)\leq 2\delta,$$
    $$\Vert A'_2-A_2\Vert\leq \Vert X_2\Vert \Vert E\Vert (\Vert X_2\Vert + \Vert P_1\Vert)\leq 6\delta\Rightarrow \Vert A'_2\Vert\leq \Vert A_2\Vert + 6\delta\leq \gamma+6\delta.$$
    From \Cref{thm:perturbation} we have $A'_1=X_1A'_1X_1$. Therefore for any eigenvector $v$ of $A'_1$ with eigenvalue $\Lambda$, we have $A_1v = v$. Therefore
    $$\Vert (A_1-A'_1)v\Vert=|1-\lambda|\Vert v\Vert \leq \Vert A_1-A'_1\Vert\Vert v\Vert\Rightarrow |1-\lambda|\leq \Vert A_1-A'_1\Vert\leq 2\delta.$$
    Finally,
    \begin{align*}
        \Vert L_2\Vert &=\Vert -P_1X_1+X_2\Vert\leq 3,\\
        \Vert R_2\Vert &=\Vert X_1P_2+X_2P_1P_2+X_2\Vert\leq 5,\\
    \end{align*}
    proving that $\kappa\leq 16$.
\end{proof}
\noindent Note that assumption (1) of~\Cref{cor:simplified_perturbation} implies the setting of~\Cref{thm:perturbation}.

\section{The URB Framework}
\label{sec:framework}
\subsection{URB Scheme and Experiment}
Let $\mathcal C(d)$ be the set of qudit channels and $H(d)$ be the set of $d \times d$ Hermitian matrices.
\begin{dfn}[URB scheme]
A \emph{univeresal randomized benchmarking (URB) scheme} $R$ on a $d$-dimensional quantum system can be expressed as a tuple $(S, \mu, \phi, M, \rho_0)$ consisting of:
\begin{itemize}
    \item A gate set $S$, encoding the gates to be applied in the URB scheme. For sake of generality, we do not restrict $S$ to be a subset of the unitary group $SU(d)$; instead $S$ can encode any description that leads to an implementation of the gate independent of other gates in a circuit.
    \item A probability distribution $\mu$ over the gate set $S$.
    \item An implementation map $\phi:S\rightarrow \mathcal{C}(d)$, assuming a gate-dependent yet Markovian noise model. 
    \item A post-processing POVM $M:S^*\rightarrow H(d)$ taking a finite string of elements from the set to a Hermitian operator on $\mathbb{C}^d$.
    \item An initial state $\rho_0\in H(d)$.
\end{itemize}
\end{dfn}
\begin{dfn}[URB experiment]
A URB scheme $R=(S,\mu,\phi,M,\rho_0)$ gives rise to the following experiment protocol:
\begin{enumerate}
    \item Given sequence length $m$, choose $m$ random elements $g_1,\cdots, g_m\in S$ i.i.d.\ according to the probability distribution $\mu$.
    \item Apply $\phi(g_1),\cdots, \phi(g_m)$ sequentially on a prepared initial state $\rho_0$.
    \item Perform a measurement $M(g_1,\cdots, g_m)$ on the final state and get a binary result.
    \item Repeat steps 1 to 3 to get an estimation $\hat{p}(m)$ of the success probability.
    \item Repeat steps 1 to 4 over appropriately chosen length parameters $m_1,\cdots, m_k$ to get estimations $\hat{p}(m_1),\cdots,\hat{p}(m_k)$. Return the estimations.
\end{enumerate}
\end{dfn}

It is easy to see that each $\hat{p}(m)$ is an unbiased estimator of the quantity
$$p_R(m):=\mathbb{E}_{g_1,\cdots, g_m
\sim\mu}\tr[M(g_1,\cdots, g_m)\cdot (\phi(g_m)\circ\cdots\circ\phi(g_1))(\rho_0)].$$
The output estimations $\hat{p}(m_1),\cdots, \hat{p}(m_k)$ are in practice assumed to be sufficiently close to $p_R(m_1),\cdots, p_R(m_k)$ and are fit to a single exponential decay curve $f(m) = A+B\cdot u^m$ to extract the decay rate $u$. The decay rate $u$ is commonly believed to indicate an ``average fidelity'' of the gate ensemble (see~\Cref{subsec:decay_exp} for the interpretation of the RB value). However, the above URB framework itself does not guarantee a single exponential decay. To guarantee such a decay, we need the following.

\begin{dfn}
A URB scheme is called an \emph{$(\epsilon, \delta, \gamma)$ twirling scheme} if the following hold:
\begin{description}
\item[Approximate factoring of the post-processing POVM into a triple $(M_0,\phi^*,\mathcal{I})$]: a \emph{final measurement} $M_0\in H(d)$, an \emph{inverting map} $\phi^*:S\rightarrow \mathcal{C}(d)$ and an intermediate channel $\mathcal{I}\in\mathcal{C}(d)$ such that the post-processing POVM can be approximately factored into three parts:
$$\sup_{m\in\mathbb{N}}\mathbb{E}_{g_1,\cdots, g_m\sim \mu}\Vert M(g_1,\cdots, g_m) - M_0\cdot \phi^*(g_1)\circ\cdots\circ \phi^*(g_m)\circ \mathcal{I}\Vert_{\infty}\leq \epsilon,$$
where we use the notation for a Hermitian matrix $\rho$ and real superoperator $\mathcal C$,
\begin{align*}
    \rho \cdot \mathcal C := \mathcal C^\dagger(\rho),
\end{align*}
$\mathcal C^\dagger$ being the adjoint superoperator with respect to the Hilbert-Schmidt norm.\footnote{We will use $\cdot$ to denote scalar multiplication, matrix multiplication, or the above shorthand. These can be differentiated by context. }
\item[Near-ideal implementation] There exists an \emph{ideal map} $\omega: S\rightarrow \mathcal{C}(d)$ into unitary channels, such that the implementation map is close to the ideal map and the inverting map is close to its adjoint:
$$\mathbb{E}_{g\sim \mu}\left[\Vert\phi(g) - \omega(g)\Vert_\diamond+\Vert\phi^*(g) - \omega(g)^\dagger\Vert_\diamond\right]\leq \delta.$$
\item[Approximate twirling] The \emph{twirling map} 
$$\Lambda^*_R:\mathcal{N}\mapsto \int_{g\in \mu}dg \omega(g)^\dagger\circ \mathcal{N}\circ \omega(g)$$
is a \emph{$\gamma$-approximate twirl}: that
$$|||\Lambda^*_R - \Lambda^*|||\leq \gamma$$
under a certain norm $|||\cdot|||$ (usually taken as the diamond norm)
where $\Lambda^* :\mathcal{N}\mapsto \int_{g\sim\eta}dg \tilde{g}^\dagger\circ \mathcal{N}\circ\tilde{g}$ is the twirling map where $\eta$ is the Haar random distribution on $SU(d)$, and $\tilde{g}:\rho\mapsto g\rho g^{-1}$ is the unitary channel corresponding to the fundamental representation of $SU(d)$ for $g$.
\end{description}
\end{dfn}
We will also say that a URB scheme $R$ is an $(\epsilon, \delta, \gamma)$ twirling scheme with respect to the tuple $(M_0,\phi^*,\mathcal{I},\omega)$ in cases where the components are to be specified. The approximately factorizable post-processing POVM condition expresses that a twirling scheme is an RB-like scheme in that after the random gates, a series of inverting gates are applied followed by a measurement. However, this mathematical definition can be more inclusive than it may seem, as we will see that linear XEB also effectively uses a factorizable post-processing POVM. The near-ideal implementation condition simply expresses that the noise levels are bounded. The most interesting condition is the $\gamma$-approximate twirl, which expresses that our random gate distribution is constant distance from a unitary 2-design.~\Cref{fig:urb} visually illustrates the URB framework and twirling schemes. 

\begin{figure}
\centering
    \begin{subfigure}{0.45\textwidth}
    \centering
    \includegraphics[width=\textwidth,trim={0 1cm 0 0}]{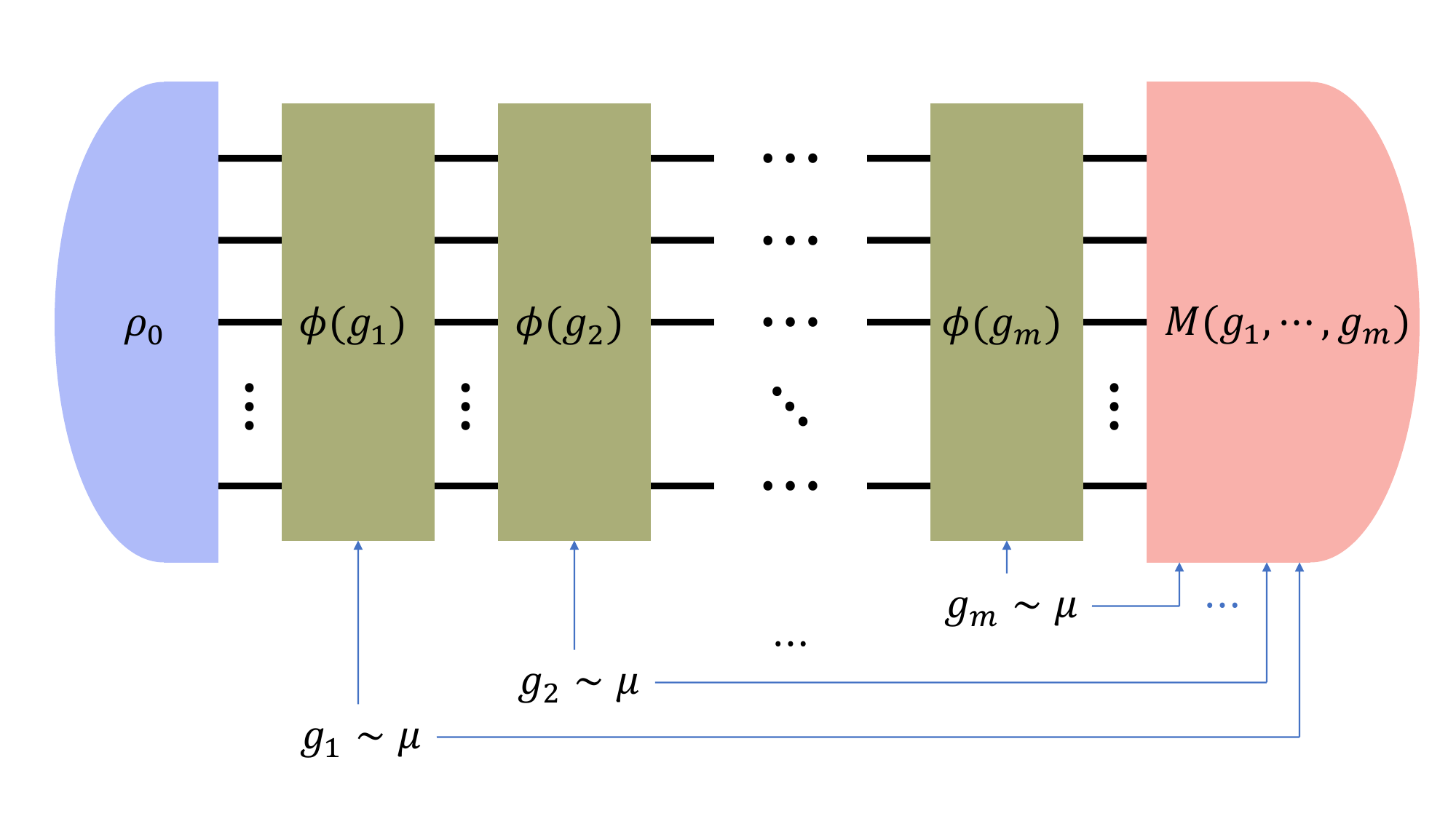}
    \caption{}
    \end{subfigure}
    \begin{subfigure}{0.45\textwidth}
    \centering
    \includegraphics[width=\textwidth,trim={0 1cm 0 0}]{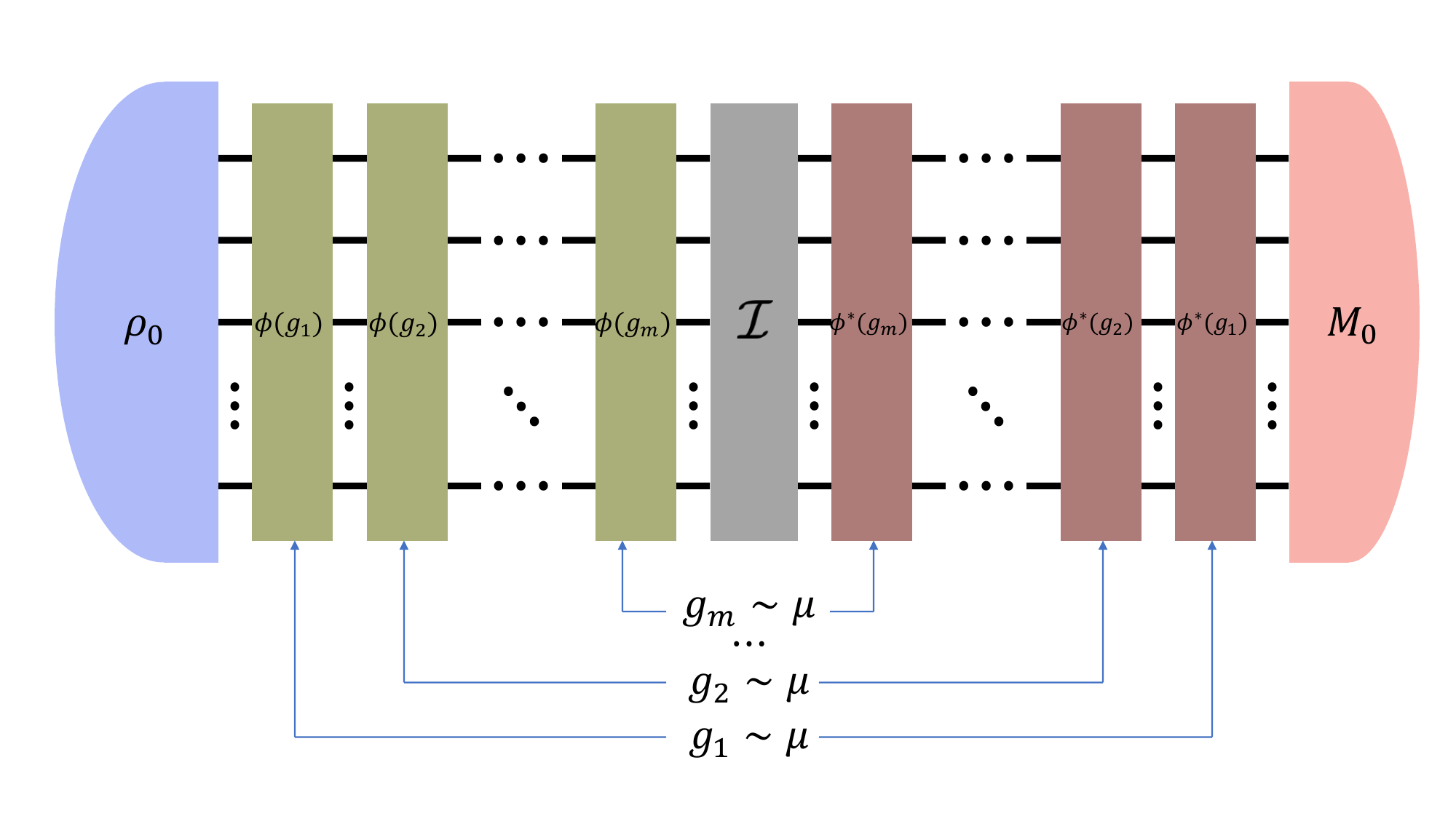}
    \caption{}
    \end{subfigure}
    
      \begin{subfigure}{0.45\textwidth}
    \centering
    \includegraphics[width=\textwidth,trim={0 5cm 0 0}]{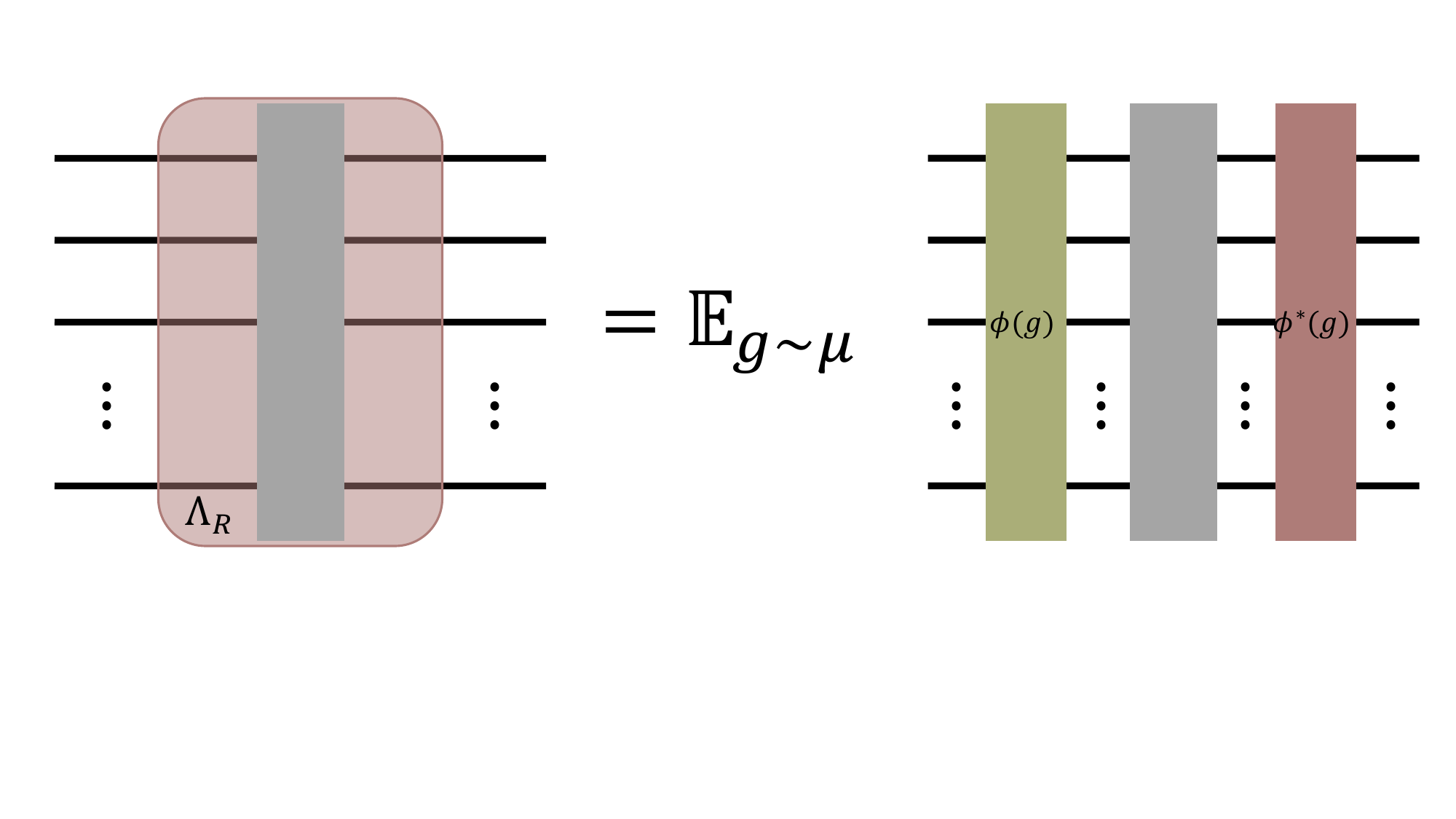}
    \caption{}
    \end{subfigure}  
    \begin{subfigure}{0.45\textwidth}
    \centering
    \includegraphics[width=\textwidth,trim={0 5cm 0 0}]{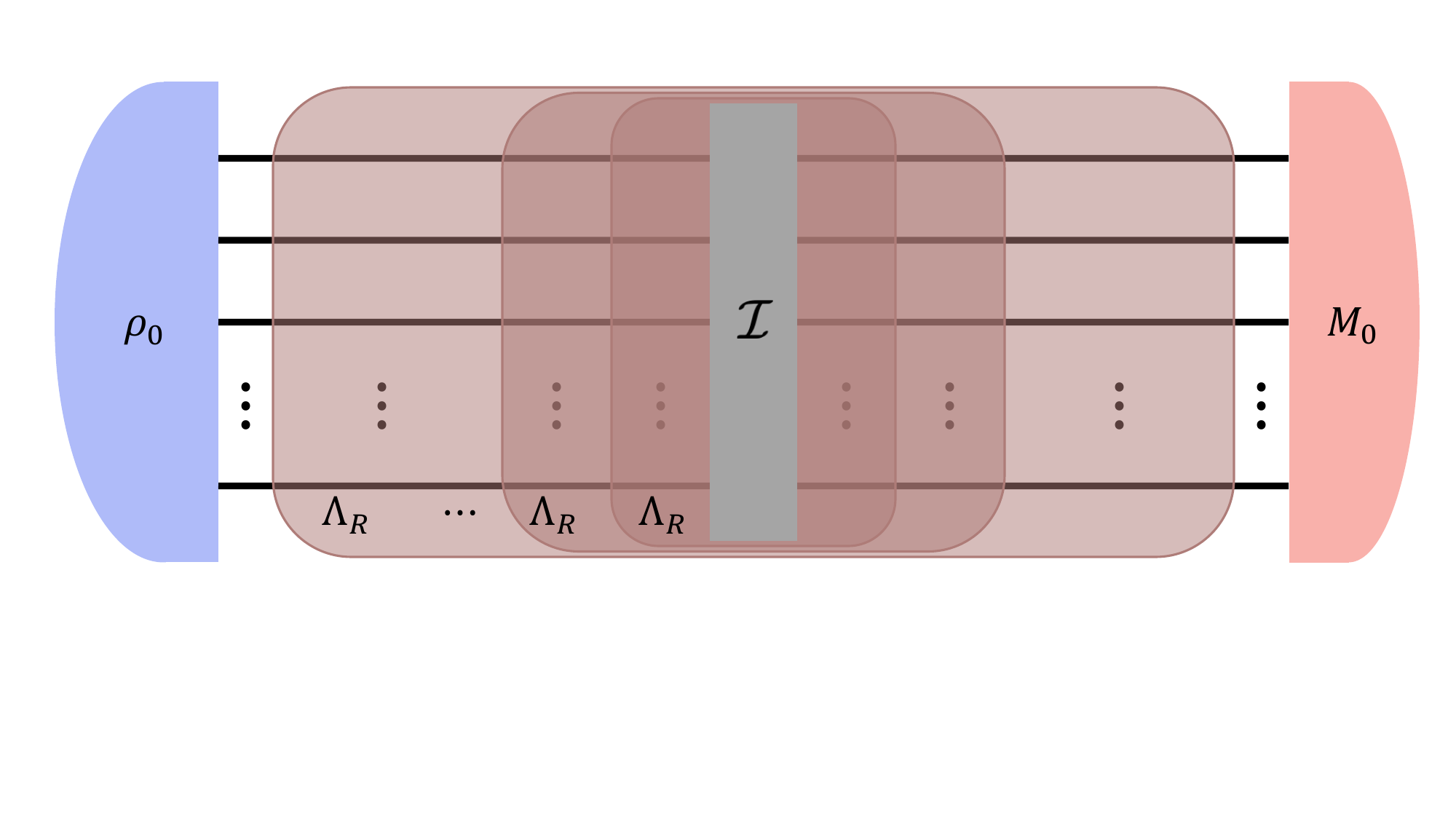}
    \caption{}
    \end{subfigure}
    \caption{A visual illustration of the URB framework and twirling schemes. (a) A generic URB scheme. Given a sequence length $m$, we randomly select $m$ i.i.d.\ samples from a distribution $\mu$. The experiment then applies $g_1,\cdots, g_m$ sequentially via an implementation map $\phi$ on a fixed initial state $\rho_0$, followed by a post-processing POVM $M$ that depends on the random elements $g_1,\cdots, g_m$. (b) Approximate factoring. When the post-processing POVM admits an approximate factoring, it can be approximately decomposed into an intermediate channel $\mathcal{I}$, a sequence of inverting operations $\phi^*(g_m),\cdots, \phi^*(g_1)$, and a fixed measurement operator $M_0$. (c) A twirling map. The average joint operation of $\phi(g)$ and $\phi^*(g)$ sandwiching a channel can be formulated as a linear operator $\Lambda_R$ on quantum channels. (d) A 
   URB scheme in terms of twirling maps. Expressing the probability of measurement in terms of the twirling map $\Lambda_R$ reduces the proof of the exponential decay to the study of its spectral properties.}
    \label{fig:urb}
\end{figure}

The main technical contribution of our work is a characterization of the exponential decay behavior provided that the URB scheme parameters satisfy certain constraints. More specifically we have the following main result.
\begin{thm}[\Cref{thm:main}, informal]
\label{thm:main_informal}
Let $R$ be an $(\epsilon, \delta, \gamma)$ twirling scheme with respect to the diamond norm, and assume that $\gamma \leq 1-11\delta$. Then there exists $A,B,p\in\mathbb{R}, p\in[1-2\delta, 1]$ such that
$$|p_R(m) - (A+B\cdot p^m)|\leq \epsilon + 16(\gamma+6\delta)^m.$$
\end{thm}
\noindent We see that the crucial property we need to establish to apply~\Cref{thm:main_informal} is $\gamma<1$, so that for sufficiently low experimental error $\delta$, $\gamma\leq 1-11\delta$. One may mistakenly think that our result is essentially saying there is a single exponential decay if our twirling map is close to the Haar twirl $\Lambda^*$, which intuitively means the unitary ensemble defined by $\mu, \omega$ is close to a unitary 2-design. This is of course unsurprising. However, we stress that we do not require $\gamma$ to be small, \emph{but just less than 1}. This is not as strong of a requirement as being close to a unitary 2-design. 

Note also that~\Cref{thm:main_informal} by itself is not sufficient to imply we can extract a single exponential decay from measured data. In general, we require the magnitude of $B$ to be significantly larger than $0$, $p$ to be significantly larger than $\gamma+6\delta$, and $\epsilon$ is sufficiently small for the URB experiment to be able to extract a single exponential decay with decay rate close to $p$ given sufficiently many repeated experiments. For the effect of $\epsilon$, see~\Cref{subsec:robustness} for an analysis of the robustness of fitting to a perturbed exponential decay.

\subsection{Examples of URB Schemes}
In the URB framework, the post-processing POVM $M$ is defined abstractly for generality. To our knowledge, all known URB schemes can be approximately factorized, but we leave open other possibilities. We here give examples of schemes that fall into our framework. Note that we make the distinction between a scheme falling into our framework and it being guaranteed by our theorem to have a single exponential decay, which requires additional assumptions. Further assumptions are required for this exponential decay to be extractable. We summarize the relationship between the URB framework, our class of twirling schemes, and the existing group-based framework as well as other classes of RB schemes in~\Cref{fig:classification}.

\begin{figure}
    \centering
    \includegraphics[width=0.55\textwidth]{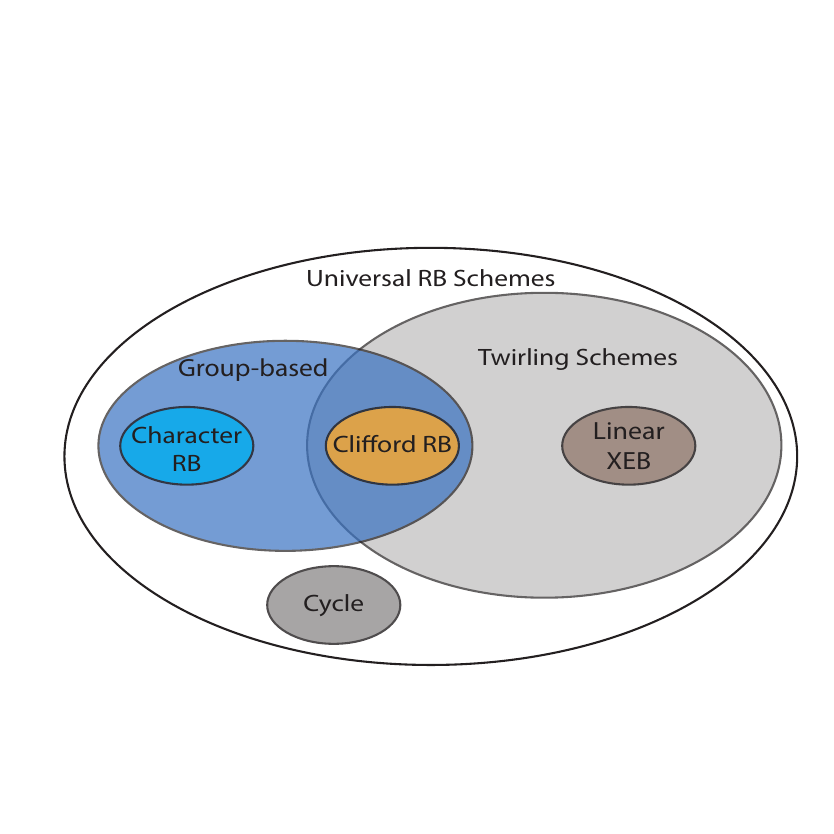}
    \caption{A diagram showing the relationships between different classes of RB schemes, including our URB framework and class of twirling schemes. Character RB refers to the class of schemes in~\cite{helsen2019new}. The space made by the URB and twirling classes leaves room for potentially new schemes. }
    \label{fig:classification}
\end{figure}

    \paragraph{Group-based Randomized Benchmarking} Standard group-based RB can be readily formulated as a URB scheme, where the gate set is taken as a group $\mathbb{G}$, with the distribution being  the uniform distribution over the group. The post-processing POVM is then $M(g_1,\cdots,g_m)=M_0\cdot \phi(g_1^{-1}\cdot\cdots\cdot g_m^{-1})$, i.e.\ applying a fixed measurement after physically applying the gate corresponding to the inverse of the product of the previous elements. This post-processing POVM admits a $\delta$-approximate factoring into the triple $(M_0,\tilde{\omega}^\dagger,\mathrm{id})$ the  when the implementation map $\phi$ is $\delta$-close to a representation $\tilde \omega$, that is, $$\mathbb{E}_{g\sim\mu}\Vert \phi(g)-\tilde{\omega}(g)\Vert_\diamond\leq \delta.$$
    Moreover, in the case that the uniform distribution over the image of the ideal map forms a unitary 2-design, the twirling map is an exact twirl, and a single exponential decay occurs whenever $\delta <\frac1{11}$. More general cases where there are different representations or multiplicities of representations still lie in the URB framework, but we do not consider them below.
    
    Note that unlike the results for example in~\cite{helsen2020general}, we have an additional constant error term $\epsilon$. This is because we want to reduce our analysis to twirling maps, and to do this we need to consider $\phi^*(g_1) \circ \cdots \circ \phi^*(g_m)$ instead of $\phi((g_m\cdots g_1)^{-1})$. We also cannot consider the latter option because we lack group structure.
     
    \paragraph{Non-uniform RB} There have been several extensions to standard RB~\cite{knill2008randomized,francca2018approximate, proctor2019direct} where the random gates are drawn from non-uniform distributions over a group, since a typical Haar random element is too costly. Such RB variants readily fit into our URB framework as it makes no assumption on the distribution. In fact, URB schemes with an ideal reference map $\omega$ naturally gives rise to a (not necessarily uniform) distribution on the special unitary group of corresponding dimension. 
    
  Our work extends~\cite{francca2018approximate} in several ways, but most importantly we do not require the distribution to be approximately uniform, nor inverse-symmetric, nor that its support contains the generators of a group. To be a twirling scheme, we only require that the twirling map is close to Haar under a certain norm. As a result, our result provides tighter bounds with less assumptions on the distribution, and also applies to infinite groups. We leave it to future work to study a twirling scheme where there are multiple irreducible representations or irreducible representations with multiplicities, which would result in a matrix exponential decay.
    
    \paragraph{Linear Cross-entropy Benchmarking (Linear XEB)} Linear cross-entropy benchmarking (Linear XEB) is a benchmarking scheme first introduced in~\cite{arute2019quantum}. In this framework, a certain number of layers, typically shallow, of random circuits $C_1,\cdots,C_m$  are applied to an initial state. A subsequent measurement then returns a bitstring $x$, whose ideal probability $q(x)$ is numerically simulated on a classical computer. The protocol returns the value $2^n q(x)-1$, where $n$ is the number of qubits in the system.
    
    We claim that linear XEB falls into the URB framework, provided that the random shallow circuits are chosen i.i.d.\ from a distribution $\mu$ over a set $S$. To see this, we express the expected outcome 
    \begin{align}
    \label{eq:linear_xeb}
    F(m)=\mathbb{E}_{C_1,\cdots, C_m\sim \mu}\mathbb{E}_{x\sim\tilde{q}}\left[2^nq(x)-1\right],
    \end{align}
    where the distributions $q,\tilde{q}$ are defined as
    $$q(x)=|\langle x| C_m\cdot\cdots\cdot C_1|0\rangle|^2,$$
    $$\tilde{q}(x)=\tr\left[M_x\cdot \phi(C_m)\circ\cdots\circ \phi(C_1)(\rho_0) \right]$$
    for some POVM $\{M_x\}_{x\in\{0,1\}^n}$ and initial state $\rho_0$.
    By lifting the distributions as diagonal operators, we have
    $$F(m)=2^n\cdot\mathbb{E}_{C_1,\cdots, C_m\sim \mu}\tr[Q\tilde{Q}]-1,$$
    where
    $$Q=\sum_x q(x)|x\rangle\langle x| =\mathcal{D}\circ\mathcal{C}_m\circ\cdots \circ\mathcal{C}_1(|0\rangle\langle 0|),$$
    $$\tilde{Q}=\sum_x \tilde{q}(x)|x\rangle\langle x| =\tilde{\mathcal{D}}\circ\phi(C_m)\circ\cdots \circ\phi(C_1)(\rho_0),$$
    with $\mathcal{D}:\rho\mapsto \sum_x\langle x|\rho|x\rangle\cdot|x\rangle\langle x|$ and $\tilde{\mathcal{D}}:\rho\mapsto \sum_x \tr[M_x\rho]\cdot|x\rangle\langle x|$ representing the ideal and physical measurement under the computational basis, and $\mathcal{C}_i$ the ideal implementation of $C_i$ as a unitary channel. Using the fact that $Q=Q^\dagger$, we have
    $$F(m)=2^n \cdot \mathbb{E}_{C_1,\cdots, C_m\sim \mu}\tr[|0\rangle\langle 0|\cdot \mathcal{C}_1^\dagger\circ\cdot \mathcal{C}_m^\dagger\circ \tilde{\mathcal{D}}\circ \phi(C_m)\circ\cdots \circ\phi(C_1)(\rho_0)]-1,$$
    that is, the post-processing POVM, being the combination of the quantum measurement followed by classical simulation, can be exactly factored in terms of a tuple $(|0\rangle\langle 0|, \tilde{\omega}^\dagger, \tilde{\mathcal{D}})$, with a prefactor of $2^n$.
    
    Linear XEB experiments suggest that $F(m)$ can be fit into a single exponential decay $ A'+B'\cdot p^m$ with $A'\approx 0$ and $B'=O(1)$. Thus, $B = 2^{-n} B' \ll 1$. This implies that the single exponential decay may not be observable unless the circuit depth approaches $\omega(n\log(1-\gamma))$ so that the error term is negligible compared to the actual single exponential decay. We leave a more detailed analysis of linear XEB experiments for future work.

    \paragraph{Cycle Benchmarking} Cycle benchmarking~\cite{erhard2019characterizing} is a specialized protocol to measure the Pauli errors of Clifford gadgets in a large-scale quantum processor. Given an $n$-qubit Clifford gate $g\in \mathcal{C}_n$, let $k = \mathrm{ord}(g) := \min\{k\in\mathbb{Z}_+|g^k=I\}$. Each random element is chosen uniformly from the set of $n$-qubit Pauli operators $\mathcal{P}_n^{\otimes k}$, and is implemented by
    $$\phi((P_1,\cdots, P_k)):=\varphi(P_1)\circ \varphi(g)\circ\cdots\circ \varphi(P_k)\circ \varphi(g),$$
    for some physical implementation $\varphi$ defined on all Pauli gates and $g$. It can be proven that in the noiseless case such implementations result in a Pauli operator, and it is therefore sufficient to calculate and implement the final Pauli gate as the recovery gate. The post-processing POVM can be approximately factored given $\varphi$ on the set of Pauli operators is close to ideal. 
    
    One distinct feature of the cycle benchmarking is that there is almost never a single exponential decay: the twirling map would be far from an approximate twirl since the distribution effectively defined on the set of Pauli channels. Instead, there is typically an exponential number of exponential decay components from a single experiment. Multiple experiments, typically with different measurements, are required in order to isolate the exponential decay components to give useful information about the Pauli error channels.

\section{Main Result}
\label{sec:result}
\subsection{Proof of Exponential Decay under Approximate Twirl}
\begin{thm}
\label{thm:main}
Let $R=(S,\mu,\phi,M,\rho_0)$ be an $(\epsilon,\delta,\gamma)$ twirling scheme with respect to the tuple $(M_0,\phi^*,\mathcal{I},\omega)$.  
Given that $\delta \leq\frac{1-\gamma}{11}$, there exists $A, B \in \mathbb{R}$ and $p \in [1-2\delta,1]$\footnote{This bound can be further tightened to $p\in[1-\delta, 1]$ using the Bauer-Fike theorem, but we leave optimization of constants in our results for sake of conciseness.} such that
\begin{align*}
    \vert p_R(m) - (A + B p^m) \vert \le \epsilon + 16(\gamma+6\delta)^m.
\end{align*}
\end{thm}
\noindent Our proof centers around \emph{twirling maps} that map channels to channels; for a URB scheme $R=(S,\mu,\phi,M,\rho_0)$ that is a twirling scheme with parameters $(M_0,\phi^*,\mathcal{I},\omega)$, define the \emph{ideal twirling map}

$$\Lambda^*_R:\mathcal{N}\mapsto \int_{g\sim \mu}dg \omega(g)^\dagger\circ \mathcal{N}\circ \omega(g)$$
and the \emph{physical twirling map}
$$\Lambda_R:\mathcal{N}\mapsto \int_{g\sim \mu}dg \phi^*(g)\circ \mathcal{N}\circ \phi(g).$$
The following is a summary of the proof: We first approximate $p_R(m)$ as a linear function of the $m$-th power of the physical twirling map $\Lambda_R$, transforming the problem into the study of the major spectral components of $\Lambda_R$. This is then analyzed by regarding $\Lambda_R$ as a perturbed version of the ideal twirling map $\Lambda^*_R$, whose spectral properties are given by the $\gamma$-approximation with respect to the Haar twirl $\Lambda^*$. For sake of simplicity we denote the Haar random distribution $\eta$ without explicity specifying the dependence on the dimension $d$.
\begin{proof}
    We first compute $p_R(m)$:
    \begin{align*}
        p_R(m) & = \mathbb{E}_{g_1, \cdots, g_{m} \sim \mu}  \tr[M(g_1,\cdots,g_m) \cdot \phi(g_m) \circ \cdots \circ \phi(g_1) (\rho_0)],
    \end{align*}
    Then,
    \begin{align}
        &\big\vert p_R(m)- \tr[M_0\cdot\Lambda_R^{m}(\mathcal{I}) (\rho_0)]\big \vert  \nonumber\\
        & = \Big \vert \mathbb{E}_{g_1, \cdots, g_{m} \sim \mu} \tr[M(g_1,\cdots, g_m) \cdot \phi(g_{m}) \circ \cdots \circ \phi(g_1) (\rho_0)] \nonumber\\
        & - \mathbb{E}_{g_1, \cdots, g_{m} \sim \mu}\tr[M_0 \cdot \phi^*(g_1) \circ \cdots \circ \phi^*(g_{m}) \circ\mathcal{I}\circ \phi(g_{m}) \circ \cdots \circ \phi(g_1) (\rho_0)] \Big\vert \nonumber\\
        & = \Big \vert \mathbb{E}_{g_1, \cdots, g_{m} \sim \mu}  \tr[\big(M(g_1,\cdots, g_m) - M_0 \cdot \phi^*(g_1) \circ \cdots \circ \phi^*(g_{m}) \circ\mathcal{I}\big) (\phi(g_{m}) \circ \cdots \circ \phi(g_1) (\rho_0))] \Big\vert \nonumber\\
        & \leq  \mathbb{E}_{g_1, \cdots, g_{m} \sim \mu}  \Big \vert\tr[\big(M(g_1,\cdots, g_m) - M_0 \cdot \phi^*(g_1) \circ \cdots \circ \phi^*(g_{m}) \circ\mathcal{I}\big)(\phi(g_{m}) \circ \cdots \circ \phi(g_1) (\rho_0))] \Big\vert \nonumber\\
        & \leq  \mathbb{E}_{g_1, \cdots, g_{m} \sim \mu} \Vert M(g_1,\cdots, g_m) - M_0 \cdot \phi^*(g_1) \circ \cdots \circ \phi^*(g_{m}) \circ\mathcal{I}\Vert_{\infty}\Vert\phi(g_{m}) \circ \cdots \circ \phi(g_1) (\rho_0)] \Vert_{\tr} \nonumber\\& =  \mathbb{E}_{g_1, \cdots, g_{m} \sim \mu} \Vert M(g_1,\cdots, g_m) - M_0 \cdot \phi^*(g_1) \circ \cdots \circ \phi^*(g_{m}) \circ\mathcal{I}\Vert_{\infty} \nonumber\\
        & \le \epsilon, \label{eq:probability_to_exp}
    \end{align}
    where the first inequality uses the triangle inequality, the second inequality uses the H\"older inequality, the third equality uses the fact that all $\phi(g_i)$'s are channel and thus $\phi(g_m)\circ\cdots\circ\phi(g_1)(\rho_0)$ is a state with unit trace, and the final inequality follows by the factorization assumption.
    
    Now we turn to study the quantity $\tr[M_0\cdot \Lambda_R^m(\mathcal{I})(\rho_0)]$. We do this by studying the spectral properties of $\Lambda_R$. Specifically, we would like to invoke \Cref{cor:simplified_perturbation} with the following setup: 
    $$\mathcal{H}=V(d), A = \Lambda^*_R, X_1 = \Lambda^*, X_2 = I -\Lambda^*, E = \Lambda_R - \Lambda_R^*,\Vert\cdot\Vert = |||\cdot|||_\diamond.$$
    We verify that all the assumptions in~\Cref{cor:simplified_perturbation} hold. 
    \begin{enumerate}
        \item $X_1^\dagger = X_1,X_1^2=X_1$. The first is because for any $\mathcal{C},\mathcal{D}\in V(d)$,
        \begin{align*}
            \langle \mathcal{C},\Lambda^*(\mathcal{D})\rangle_{SO}&=\langle \mathcal{C}, \int_{g\sim \eta} dg\tilde{g}^\dagger\circ \mathcal{D}\circ\tilde{g} \rangle_{SO}\\
            &=\int_{g\sim \eta} dg \langle \mathcal{C}, \tilde{g}^\dagger\circ \mathcal{D}\circ\tilde{g} \rangle_{SO}\\
            &=\int_{g\sim \eta} dg \sum_i\langle \mathcal{C}(Y_i), \tilde{g}^\dagger\circ \mathcal{D}\circ\tilde{g}(Y_i) \rangle_{HS}\\
            &=\int_{g\sim \eta} dg \sum_i\langle \tilde{g}\circ\mathcal{C}(Y_i),  \mathcal{D}\circ\tilde{g}(Y_i) \rangle_{HS}\\
            &=\int_{g\sim \eta} dg \sum_i\langle \tilde{g}\circ\mathcal{C}\circ\tilde{g}^\dagger(Y^g_i),  \mathcal{D}(Y^g_i) \rangle_{HS}\\
            &=\int_{g\sim \eta} dg \langle \tilde{g}\circ\mathcal{C}\circ\tilde{g}^\dagger,  \mathcal{D} \rangle_{SO}\\
            &= \langle\int_{g\sim \eta} dg \tilde{g}\circ\mathcal{C}\circ\tilde{g}^\dagger,  \mathcal{D} \rangle_{SO}\\
            &= \langle\int_{g^\dagger\sim \eta} dg \tilde{g}\circ\mathcal{C}\circ\tilde{g}^\dagger,  \mathcal{D} \rangle_{SO}\\
            &= \langle\Lambda^*(\mathcal{C}),  \mathcal{D} \rangle_{SO}.
        \end{align*}
        Here we used the facts that $\{Y_i^g\}_i:=\{\tilde{g}(Y_i)\}_i$ is an orthonormal basis under the Hilbert-Schmidt inner product if $\{Y_i\}_i$ is, that $(\tilde{g}^\dagger)^\dagger = \tilde{g}$, and that the Haar measure is inverse symmetric. 
        
        To prove $X_1^2=X_1$, define $\Lambda(\nu):=\int_{g\sim\nu}dg \tilde{g}^\dagger\circ\cdot\circ\tilde{g}$ for arbitrary distribution $\nu$ defined on $SU(d)$. We then have
        \begin{align*}
            \Lambda(\nu)\Lambda^*&=\int_{g\sim \nu}dg\int_{h\sim \eta}dh \tilde{g}^\dagger\circ\tilde{h}^\dagger\cdot \tilde{h}\circ\tilde{g}\\
            &=\int_{g\sim \nu}dg\int_{h\sim \eta}dh \tilde{(hg)}^\dagger\cdot \tilde{(hg)}\\
            &=\int_{g\sim \nu}dg\int_{h g\sim \eta}d(h g) \tilde{(h g)}^\dagger\cdot \tilde{(h g)}\\
            &=\int_{g\sim \nu}dg\int_{h'\sim \eta}dh' \tilde{h'}^\dagger\cdot \tilde{h'}\\
            &=\int_{g\sim \nu}dg\Lambda^*=\Lambda^*,
        \end{align*}
        where we have used the fact that the Haar measure is right invariant. Taking $\nu=\eta$ proves $X_1^2=X_1$.
        \item $X_iAX_j=A_i\delta_{ij}$. Since $\Lambda^*_R$ is uniquely determined by the distribution $\mu_\omega$ on $SU(d)$ which is defined by the distribution $\mu$ on $S$ and the ideal map $\omega$, we have
        $\Lambda^*\Lambda^*_R=\Lambda^*$ and similarly $\Lambda^*_R\Lambda^*=\Lambda^*$ from the above. This directly leads to
        $$X_1AX_1=\Lambda^*=A_1, X_1AX_2=\Lambda^*\Lambda^*_R(I-\Lambda^*)=\Lambda^*(I-\Lambda^*)=\Lambda^*-\Lambda^*=0,$$
        $$X_2AX_1=0,X_2AX_2=(I-\Lambda^*)\Lambda^*_R(I-\Lambda^*)=\Lambda^*_R-\Lambda^*=A_2.$$
        \item $\Vert A_2\Vert = ||| \Lambda^*_R-\Lambda^*|||_\diamond\leq \gamma$ by $\Lambda^*_R$ being a $\gamma$-approximated twirl.
        \item $\Vert E\Vert \leq\delta$. This is because
        \begin{align}
        \Vert E\Vert &= |||\Lambda_R-\Lambda^*_R|||_\diamond\nonumber\\
        &=||| \int_{g\sim \mu}dg\big(\omega(g)^\dagger\circ\cdot\circ \omega(g) - \phi^*(g)\circ\cdot\circ \phi(g)\big)|||_\diamond\nonumber\\
        &=||| \int_{g\sim \mu}dg\big((\omega(g)^\dagger-\phi^*(g))\circ\cdot\circ \omega(g) + \phi^*(g)\circ\cdot\circ (\omega(g)-\phi(g))\big)|||_\diamond\nonumber\\
        &\leq\int_{g\sim \mu}dg\Vert \omega(g)^\dagger-\phi^*(g)\Vert_\diamond\Vert \omega(g)\Vert_\diamond + \int_{g\sim \mu}dg\Vert\phi^*(g)\Vert_\diamond\Vert \omega(g)-\phi(g)\Vert_\diamond\nonumber\\
        &\leq\int_{g\sim \mu}dg\Vert \omega(g)^\dagger-\phi^*(g)\Vert_\diamond + \int_{g\sim \mu}dg\Vert \omega(g)-\phi(g))\Vert_\diamond\nonumber\\
        &\leq \delta,\label{eqn:enorm}
        \end{align}
        where the first inequality uses the triangle inequality and \Cref{prp:diamond_norm_bound}, and the second inequality uses that $\Vert\omega(g)\Vert_\diamond,\Vert \phi^*(g)\Vert_\diamond= 1$ since they are channels. 
        \item $A_1=X_1$ is already proven above.
        \item $\Vert X_1 \Vert_\diamond \leq 1$ follows directly from~\Cref{prp:diamond_norm_bound}.
        \item This is assumed.
    \end{enumerate}
    We can now invoke~\Cref{cor:simplified_perturbation} to obtain $L_1,L_2,R_1,R_2, A_1', A_2'$ such that $\Lambda_R$ is diagonalized:
    $$\Lambda_R = L_1 A'_1 R_1+L_2A'_2R_2,$$
    and $\Vert A'_1-X_1\Vert\leq 2\delta, \Vert A'_2\Vert\leq \gamma+6\delta$, and $\Vert L_2\Vert\Vert R_2\Vert\leq 16$.
    
    The diagonalization ensures that $\Lambda_R^m = L_1 (A'_1)^m R_1+L_2(A'_2)^mR_2$ and therefore
    \begin{align*}
        \tr[M_0\cdot \Lambda^m_R(\mathcal{I})(\rho_0)]&=\tr[M_0\cdot (L_1 (A'_1)^m R_1+L_2(A'_2)^mR_2)(\mathcal{I})(\rho_0)]\\
        & = \tr[M_0\cdot L_1 (A'_1)^m R_1(\mathcal{I})(\rho_0)]+\tr[M_0\cdot L_2(A'_2)^mR_2(\mathcal{I})(\rho_0)].
    \end{align*}
    The second term is bounded by an exponential decay: 
    \begin{align}
        \vert \tr[M_0\cdot L_2(A'_2)^mR_2(\mathcal{I})(\rho_0)]\vert &\leq \Vert L_2(A'_2)^mR_2(\mathcal{I})\Vert_\diamond\nonumber\\
        &\leq \Vert L_2(A'_2)^mR_2\Vert\nonumber\\
        &\leq \Vert L_2\Vert\cdot\Vert(A'_2)^m\Vert\cdot\Vert R_2\Vert\nonumber\\
        &\leq 16\cdot \Vert A'_2\Vert^m\nonumber\\
        &\leq 16\cdot (\gamma+6\delta)^m\label{eqn:errnorm}.
    \end{align}
    We finally show that the term $\tr[M_0\cdot L_1 (A'_1)^m R_1(\mathcal{I})(\rho_0)]$ is exactly a single exponential decay.  Specifically, we show that the rank of $A'_1$ is 2, and it can be diagonalized with one eigenvalue 1 and the other $p\in[1-2\delta, 1]$.
    
   \begin{enumerate}
       \item \textbf{$A'_1$ is of rank 2}. First, by construction the rank is at most 2. Also, for any quantum channel $\mathcal{N}$,
       $$\Lambda^*(\mathcal{N})=\alpha_\mathcal{N} \cdot \mathrm{id} + (1-\alpha_\mathcal{N})\cdot \mathrm{dep}$$
       for some $\alpha_\mathcal{N}$. Therefore $\mathrm{Im}(\Lambda^*)=\mathrm{span}(\{\mathrm{dep},\mathrm{id}\})$ and $\mathrm{rank}(\Lambda^*)=2$. Moreover, since $\Vert A'_1-A_1\Vert\leq 2\delta$ and $A_1=\Lambda^*$ is a projector, the two eigenvalues $\lambda_1,\lambda_2$ of $A'_1$ must be $2\delta$-close to $1$ and $A'_1$ is therefore of rank $2$. 
       \item \textbf{Eigenvalues of $A'_1$ are bounded in magnitude by $1$.} For any non-zero eigen-channel $\mathcal{N}$ of $\Lambda_R$ with eigenvalue $\lambda$, we have
       $$|\lambda|\Vert\mathcal{N}\Vert_\diamond=\Vert \Lambda_R(\mathcal{N})\Vert_\diamond\leq ||| \Lambda_R|||_\diamond\cdot\Vert\mathcal{N}\Vert_\diamond\Rightarrow |\lambda|\leq ||| \Lambda_R|||_\diamond\leq 1.$$
       Therefore all eigenvalues of $\Lambda_R$ are bounded in magnitude by $1$, and so are $A'_1, A'_2$ since their two spectra constitute a bipartition of the spectrum of $\Lambda_R$.
       \item \textbf{All $\Lambda_R,\Lambda^*_R$ and $\Lambda^*$ are real-valued under a properly chosen basis.} 
       Consider a orthogonal basis $B$ (under the Hilbert-Schmidt inner product) of the space of Hermitian matrices. Then, any real superoperator is defined by its action on elements of $B$, giving rise to a matrix representation. In particular, real superoperators will be real-valued matrices since they preserve hermiticity. Now, we can consider a basis $\mathcal B := \{\mathcal{N}_i\}_i$ of real superoperators with a matrix representation with respect to $B$ being an elementary matrix: a single $1$ and $0$'s everywhere else. Since $\omega(g),\phi(g)$ and $\phi^*(g)$ are quantum channels for all $g$, $\Lambda_R,\Lambda_R^*,\Lambda^*$ maps $\mathcal N_i$ to real matrices with respect to $B$ and therefore are themselves real matrices with respect to $\mathcal B$. This also indicates that $L_1,L_2,R_1,R_2,A'_1,A'_2$ are real with respect to $\mathcal B$. Additionally, the intermediate channel $\mathcal{I}$, the initial state $\rho_0$, and the final measurement $M_0$ are all real with respect to the corresponding bases.
       \item \textbf{$A'_1$ has one eigenvalue $1$.} Consider the channel
       $$\mathcal{C}_{\phi^*}:=\int_{g\sim\mu} dg\phi^*(g).$$
       By Brouwer's fixed-point theorem, there exists a quantum state $\rho^*$ such that $\mathcal{C}_{\phi^*}(\rho^*)=\rho^*$ is an eigenstate of $\mathcal{C}_{\phi^*}$ with eigenvalue $1$. $\rho^*$ can then be lifted to an eigen-channel $\mathcal{E}_{\rho^*}:\rho\mapsto \rho^*\cdot\tr[\rho]$ of $\Lambda_R$ with eigenvalue $1$. This eigenvalue must lie in $A'_1$ as $ \Vert A'_2\Vert\leq \gamma+6\delta<1$. Without loss of generality denote $\lambda_1=1$. Then $\lambda_2$ must be real as $A'_1$ is real and so eigenvalues come in complex conjugate pairs, and $\lambda_2\in[1-2\delta, 1]$.
       \item \textbf{$A'_1$ is diagonalizable even if $\lambda_2=1$.} Suppose otherwise and $A'_1 = \begin{bmatrix}
       1 & a\\0&1
       \end{bmatrix}$ under properly chosen $L'_1,R'_1$ with $a\neq 0$. Then $\tr[M_0\cdot L_1 (A'_1)^m R_1(\mathcal{I})(\rho_0)]$ becomes a linear function with respect to $m$ with non-zero coefficients for some $M_0,\mathcal{I},\rho_0$, and thus becomes unbounded. This contradicts to the fact that it is the difference of two bounded functions:
       $$\tr[M_0\cdot L_1 (A'_1)^m R_1(\mathcal{I})(\rho_0)] =\tr[M_0\cdot \Lambda^m_R(\mathcal{I})(\rho_0)]-\tr[M_0\cdot L_2(A'_2)^mR_2(\mathcal{I})(\rho_0)], $$
       where the magnitude of the first term is bounded by $p_R(m)+\epsilon\leq 1+\epsilon$ and that of the second term by $16(\gamma+6\delta)^m\leq 16$.
   \end{enumerate}
\end{proof}
\begin{rmk}
\label{rmk:difference_of_channels}
\emph{
In the proof we chose the Hilbert space $\mathcal{H}$ to be the channel space $V(d)$ as it is an invariant space under all twirling maps $\Lambda^*,\Lambda^*_R$ and $\Lambda_R$. A similar proof considers the channel difference space $V_0(d)$ as it is again an invariant subspace under all above twirling maps by taking as input the channel difference $\mathcal{I}-\mathcal{E}_{\rho^*}$. We will use this fact when giving upper bounds on $\gamma$.
}
\end{rmk}

\subsection{Proving Measures are Approximate Twirls}
\label{subsec:gamma_bounds}
In~\Cref{thm:main}, a key requirement is for the ideal twirling map $\Lambda^*_R$ to be a $\gamma$-approximate twirl with $\gamma < 1$. Here, we give several possible ways to upper bound $\gamma$. A trivial upper bound is 2 which follows directly from~\Cref{prp:diamond_norm_bound}.

We establish some notation. Note that a probabilistic distribution $\mu$ on the set $S$ together with the ideal map $\omega:S\rightarrow SU(d)$ defines a probabilistic distribution on $SU(d)$, from which the ideal twirling map $\Lambda^*_R$ is uniquely determined. For the rest of the section, we consider probabilistic distributions $\mu$ on $SU(d)$, and let 
\begin{align*}
    \Lambda(\mu) := \int_{g \sim \mu} \omega(g)^\dagger\circ \cdot \circ\omega(g) dg.
\end{align*}

\subsubsection{Bounds on $\gamma$ for Convex Combinations}
One straightforward yet somewhat restrictive way of bounding $\gamma$ is when $\mu$ is a convex combination of a unitary 2-design $\nu$ with any other measure. 


We first observe that if $\mu$ is a $p$-convex combination of some measure $\mu'$ and $\nu$, then
$$\Lambda(\mu) = \Lambda(p\mu'+(1-p)\nu) = p\Lambda(\mu') + (1-p)\Lambda(\nu).$$
Note that this implies the set of twirling maps corresponding to some measure is convex. However, this implies
$$||| \Lambda(\mu) - \Lambda(\nu) |||_\diamond = ||| p \Lambda(\mu') + (1-p) \Lambda(\nu) - \Lambda(\nu)|||_\diamond \leq p ||| \Lambda(\mu')|||_\diamond + p ||| \Lambda(\nu)|||_\diamond =2p.$$
Moreover, $\Lambda(\nu) = \Lambda(\eta)$, where $\eta$ is the Haar random distribution on $SU(d)$. We conclude the following:
\begin{prp}
If $\mu$ is a $p$-convex combination of any measure and a unitary 2-design, then $\mu$ is a $2p$-approximate twirl.
\end{prp}

Consider the Clifford group $C(d) \subseteq SU(d)$ and let the measure $\mu_C$ be the uniform distribution over $C(d)$. Suppose a discrete measure $\mu$ has support that includes $C(d)$. Then, define
$$m := \min_{c \in C(d)} \mu(c) > 0.$$
We have that
$$m \times \vert C(d) \vert \leq \sum_{c \in C(d)} \mu(c) \leq 1.$$
If $m$ satisfies this with equality then $\mu = \mu_C$. Assume it doesn't. Then, define the measure $\mu'$ as 
$$\mu' = \frac{1}{1-m\vert C(d)\vert} (\mu - m \vert C(d) \vert \mu_C).$$
By construction, $\mu'$ is a legitimate measure. That is, $\mu$ is a $(1-m \vert C(d)\vert)$-convex combination of $\mu'$ and $\mu_C$. Note that $\mu$ having support that includes $C(d)$ is also a necessary condition for it to be a nontrivial convex combination involving $\mu_C$. By the above result, we can conclude $\mu$ is a $2(1-m\vert C(d)\vert)$-approximate twirl. We can therefore conclude
\begin{prp}
If a measure $\mu$'s support includes the Clifford group with all probabilities greater than half of that of the uniform distribution over the Clifford group, it is a $\gamma$-approximate twirl with $\gamma<1$.
\end{prp}
\noindent In particular, \emph{we can find another distribution over the Clifford group for which we can do URB.} This can be extended to any unitary 2-design that is a uniform distribution over a finite set. We can also consider more general two-designs. Let $\nu$ be a measure with support a finite subset $S \subseteq SU(d)$ that is a unitary 2-design. We define
$$M_S(\nu) := \max_{g \in S} \nu(S) > 0.$$
We have
$$M_S(\nu) \times \vert S \vert \geq \sum_{g \in S} \nu(g) = 1.$$
Consider a measure $\mu$ whose support includes $S$. Define
$$m_S(\mu) := \min_{g\in S}\mu(S).$$
Now, 
$$m_S(\mu) \times \vert S \vert \le \sum_{g\in S} \mu(g) \leq 1.$$
We conclude $m_S(\mu) \leq M_S(\nu)$. Then, we can define the measure
$$\mu' := \frac{1}{1-m_S(\mu)/M_S(\nu)} \left(\mu -\frac{m_S(\mu)}{M_S(\nu)} \nu\right).$$
By construction, $\mu'$ is a legitimate measure. Thus, $\mu$ is a convex combination of $\mu'$ and $\nu$. Note that $\mu$ having support on $S$ is a necessary condition for $\mu$ to be a nontrivial convex combination involving $\nu$. 
We conclude
\begin{prp}
Given a unitary 2-design $\nu$ with finite support $S$,  if a measure $\mu$ has probability greater than half of the max probability of $\nu$ on all of $S$, $\mu$ is a $\gamma$-approximate twirl with $\gamma<1$. 
\end{prp}
\noindent Note that this implies if $M_S(\nu) \geq 2/\vert S\vert$,
$$\sum_{g\in S} \mu(g) \geq m_S(\mu) \times \vert S \vert > \frac 1 2 M_S(\nu) \times \vert S \vert \geq 1,$$
which is impossible. Thus in this case we cannot find such a measure $\mu$. Note that we can easily extend these arguments to the case when the support of $\mu$ is not finite.

\subsubsection{Bounds on $\gamma$ via Pauli Representation}
For the following bounds we restrict ourselves to the case of qubits, that is $d= 2^n$. We consider the Pauli matrices $\{P_i\}_{i=1}^{d^2}$, that is tensor products of $\{I, X, Y, Z\}$, that span $L(\mathbb{C}^d)$ and are orthogonal under the Hilbert-Schmidt inner product
\begin{align*}
    \tr[P_i P_j] = d\delta_{ij}.
\end{align*}

With this basis, we can express any linear operator. Since superoperators can be expressed in terms of left and right multiplications by operators, in general we can express a superoperator as a matrix $\alpha$ whose coefficients are defined by:
\begin{align*}
    \mathcal T = \sum_{i,j} \alpha_{i,j} P_i \cdot P_j.
\end{align*} 
We refer to this as the $\alpha$ matrix and will sometimes specify which superoperator we're considering by $\alpha(\mathcal T)$. Unless otherwise specified, summation indices range from $1$ to $d^2$. Going further, given a supersuperoperator $\Lambda$, we now introduce the $\beta$ tensor whose coefficients are defined by
\begin{align*}
    \Lambda: P_i \cdot P_j \mapsto \sum_{k,l} \beta_{i,j}^{k,l} P_k \cdot P_l.
\end{align*}
We can linearly extend this to obtain the action of $\Lambda$ on any superoperator. We refer to this as the $\beta$ tensor. The supersuperoperators we will mainly be considering are twirling maps corresponding to a measure $\mu$, and we sometimes write $\beta(\mu)$ as the $\beta$ tensor for $\Lambda(\mu)$. For a more detailed introduction to these representations and some basic facts that will be relevant to proving upper bounds on $\gamma$, see~\Cref{sec:app_alpha_beta}.

The first bound is the following:
\begin{prp}
\label{prp:l2_norm_bound}
Let $\beta(\mu)$ be the $\beta$ tensor for the twirling map $\Lambda(\mu)$. Then,
\begin{align*}
    ||| \Lambda(\mu) - \Lambda(\eta) |||_{\diamond} \leq  2 \sqrt{\sum_{\substack{k,l,k',l'=1}}^{d^2} \left\vert\sum_{i,j=1}^{d^2} (\beta^{k,l}_{i,j}(\mu)-\beta^{k,l}_{i,j}(\eta))(\beta^{k',l'}_{i,j}(\mu)-\beta^{k',l'}_{i,j}(\eta))\right\vert}.
\end{align*}
\end{prp}
\begin{proof}
    We prove this in~\Cref{sec:app_l2_bound}.
\end{proof}
\noindent This is a direct and general upper bound on the diamond norm. However, its form is reminiscent of an $\ell_2$-norm, while we would expect something that looks more like an induced $\ell_1$-norm. We can prove such a bound for a special case. 
\begin{prp}
\label{prp:induced_l1_bound}
Suppose
\begin{align*}
    \Lambda(\mu) = \Lambda(\mu) \circ \Lambda(\mu_P),
\end{align*}
where $\mu_P$ is the uniform distribution over Pauli operators. Furthermore, let $\mathcal T$ be a difference of two channels: $\mathcal T = \mathcal N_1 - \mathcal N_2$. Then,
\begin{align*}
    \Vert [\Lambda(\mu) - \Lambda(\eta)](\mathcal T)\Vert_\diamond \leq \max_{i \in [d^2]}\sum_{k,l=1}^{d^2}  \vert\beta_{i,i}^{k,l}(\mu)-\beta_{i,i}^{k,l}(\eta) \vert \Vert \mathcal T \Vert_\diamond
\end{align*}
where $[d^2] := \{1, 2, \cdots, d^2\}$.
\end{prp}
\begin{proof}
    We prove this in~\Cref{sec:app_induced_l1_bound}.
\end{proof}
\noindent Intuitively, the Pauli representation lets us express supersuperoperators as a linear map on matrices, and our upper bound on the corresponding induced diamond norm resembles the induced $\ell_1$-norm for regular matrices:
\begin{align*}
    \Vert M \Vert_{1 \to 1} = \max_j \sum_i \vert M_{i,j} \vert.
\end{align*}
The assumption that $\Lambda(\mu)$ is right invariant under $\Lambda(\mu_P)$ is natural in that this allows us to connect the trace norm of $\alpha$ with the superoperator's diamond norm. This right invariance is easily enforced by adding a uniformly random Pauli operator before each random gate for any URB scheme, along with the fact that $\Lambda(\mu_P)$ is a projector as shown in~\Cref{sec:app_induced_l1_bound}. Finally, the restriction that $\mathcal T$ is a difference of channels still allows us to bound $\gamma$ in~\Cref{thm:main}. This is because of~\Cref{rmk:difference_of_channels} and the fact that 
\begin{align*}
    V_0(d) = \{ c \cdot(\mathcal N_1-\mathcal N_2) \vert c \in \mathbb{R}, \mathcal N_i \in \mathcal{C}(d)\}.
\end{align*}
To see this, we simply observe that a linear combination of a difference of channels is simply a scaled difference of channels. Without loss of generality, let $c_1, c_2 \geq 0$ and consider
\begin{align*}
    c_1 (\mathcal N_1 - \mathcal N_2) + c_2 (\mathcal N_3 - \mathcal N_4) = (c_1+c_2) \cdot \left[\left(\frac{c_1}{c_1+c_2} \mathcal N_1 + \frac{c_2}{c_1+c_2} \mathcal N_3\right) -\left(\frac{c_1}{c_1+c_2} \mathcal N_2 + \frac{c_2}{c_1+c_2} \mathcal N_4\right)\right].
\end{align*}
As an exercise, we can evaluate $\gamma$ for $\mu_P$:
\begin{align*}
    \max_{i \in [d^2]}\sum_{k,l=1}^{d^2}  \vert\beta_{i,i}^{k,l}(\mu)-\beta_{i,i}^{k,l}(\eta) \vert  &= \max_{i>1}\sum_{k,l=1}^{d^2}  \vert\beta_{i,i}^{k,l}(\mu)-\beta_{i,i}^{k,l}(\eta) \vert \\
    & = \left\vert 1 - \frac{1}{d^2-1} \right\vert + (d^2-2)\cdot \left\vert 0- \frac{1}{d^2-1}\right\vert\\
    & = 2 \frac{d^2-2}{d^2-1}.
\end{align*}
This for $d=2$ is already $4/3>1$, so this does not justify uniform Pauli randomized benchmarking. However, as $d\to\infty$, the upper bound goes to $2$, which is the highest the induced diamond norm distance can be.

\subsection{Approximate Twirls in Other Norms}
\Cref{thm:main} shows that a single exponential decay can be observed when $R$ is a twirling scheme with respect to the induced diamond norm. We give upper bounds on the induced diamond norm, but in general this is difficult to compute. Other norms are easier to compute and characterize, and we give the following corollaries as alternatives to~\Cref{thm:main}. 
\begin{cor}[Single exponential decay with respect to the trace norm]\label{cor:iitrace}
Let $R=(S,\mu,\phi,M,\rho_0)$ be an $(\epsilon,\delta,\gamma)$ twirling scheme with respect to the tuple $(M_0,\phi^*,\mathcal{I},\omega)$ under the trace norm. 
Given that $\delta \leq\frac{1-\gamma}{11}$, there exists $A, B \in \mathbb{R}$ and $p \in [1-2\delta,1]$ such that, there exists $A, B \in \mathbb{R}$ and $p \in [1-2\delta,1]$ such that
\begin{align*}
    \vert p_R(m) - (A + B p^m) \vert \le \epsilon + 16(\gamma+6\delta)^m.
\end{align*}
\end{cor}
\begin{proof}
    This can be readily proved by replacing the diamond norms in the proof of \Cref{thm:main} by the trace norm.
\end{proof}

Given a twirling map, it is often much easier to observe its spectrum through diagonalization of its matrix representation. In contrast, the diamond norm or trace norm are typically computed via semidefinite programming. The following gives an analog of \Cref{thm:main} in terms of a spectral gap, although with an error term involving the dimension of the quantum system.
\begin{cor}[Single exponential decay with respect to the Frobenius norm]\label{cor:spectral}
Let $R=(S,\mu,\phi,M,\rho_0)$ be an $(\epsilon,\delta,\gamma)$ twirling scheme on quantum systems with dimension $d$ with respect to the tuple $(M_0,\phi^*,\mathcal{I},\omega)$ under the $\Vert \cdot \Vert_2$ norm. Let $\delta':=d\cdot \delta$.
Given that $\delta'\leq \frac{1-\gamma}{11}$, there exists $A, B \in \mathbb{R}$ and $p \in [1-2\delta',1]$ such that
\begin{align*}
    \vert p_R(m) - (A + B p^m) \vert \le \epsilon + 16d^{3/2}(\gamma+6\delta')^m.
\end{align*}
Furthermore, in the case that $\phi^*$ or $\phi$ maps to unitary mixtures we can restrict to $\delta'= \sqrt{d}\cdot\delta$.
\end{cor}
\begin{proof}
There are only two places in the proof of \Cref{thm:main} where the diamond norm is explicitly used, and we replace them with the spectral norms. The first one is at \Cref{eqn:enorm}, to bound the norm of the perturbation $E$. Here we instead have
\begin{align}
        \Vert E\Vert &= |||\Lambda_R-\Lambda^*_R|||_2\nonumber\\
        &=||| \int_{g\sim \mu}dg\big(\omega(g)^\dagger\circ\cdot\circ \omega(g) - \phi^*(g)\circ\cdot\circ \phi(g)\big)|||_2\nonumber\\
        &=||| \int_{g\sim \mu}dg\big((\omega(g)^\dagger-\phi^*(g))\circ\cdot\circ \omega(g) + \phi^*(g)\circ\cdot\circ (\omega(g)-\phi(g))\big)|||_2\nonumber\\
        &\leq\int_{g\sim \mu}dg||| (\omega(g)^\dagger-\phi^*(g))\circ\cdot\circ \omega(g)|||_2 + \int_{g\sim \mu}dg|||\phi^*(g)\circ\cdot\circ (\omega(g)-\phi(g))|||_2\nonumber\\
        &\leq\int_{g\sim \mu}dg\Vert \omega(g)^\dagger-\phi^*(g)\Vert_2\Vert \omega(g)\Vert_2 + \int_{g\sim \mu}dg\Vert\phi^*(g)\Vert_2\Vert \omega(g)-\phi(g)\Vert_2\nonumber\\
        &\leq\sqrt{d}(\int_{g\sim \mu}dg\Vert \omega(g)^\dagger-\phi^*(g)\Vert_2 + \int_{g\sim \mu}dg\Vert \omega(g)-\phi(g))\Vert_2)\nonumber\\
        &\leq d\big(\int_{g\sim \mu}dg\Vert \omega(g)^\dagger-\phi^*(g)\Vert_\diamond + \int_{g\sim \mu}dg\Vert \omega(g)-\phi(g))\Vert_\diamond\big)\nonumber\\
        &\leq d\delta,\label{eqn:enorm_spec}
        \end{align}
where in the third inequality we used the fact that  $\Vert\omega(g)\Vert_2=1\leq\sqrt{d}$ and $\phi^*(g)$ have 2 norm upper bounded by $\sqrt{d}$.  In case that $\phi^*(g)$ is a unitary mixture (similarly for $\phi(g)$), the bound can further be tightened to $\sqrt{d}\delta$ by observing $\Vert\phi^*(g)\Vert_2\leq \Vert\phi^*(g)\Vert_{\tr{}}=1$.
The second one is at \Cref{eqn:errnorm}. Here we have
\begin{align}
        \tr[M_0\cdot L_2(A'_2)^mR_2(\mathcal{I})(\rho_0)]&\leq \Vert L_2(A'_2)^mR_2(\mathcal{I})\Vert_{\tr}\nonumber\\
        &\leq ||| L_2(A'_2)^mR_2|||_{\tr}\nonumber\\
        &\leq d^{3/2}\cdot ||| L_2(A'_2)^mR_2|||_2\nonumber\\
        &\leq d^{3/2}\cdot |||L_2|||_2\cdot|||(A'_2)^m|||_2\cdot||| R_2|||_2\nonumber\\
        &\leq 16d^{3/2}\cdot \Vert A'_2\Vert_2^m\nonumber\\
        &\leq 16d^{3/2}\cdot (\gamma+6\delta')^m\label{eqn:errnorm_spec}.
    \end{align}
\end{proof}
As in the case of~\Cref{thm:main}, the crucial property we need to establish is $\gamma < 1$, but now with respect to these other norms.

\subsection{Interpretation of the Decay Rate}
\label{subsec:decay_exp}
\subsubsection{Gauge Invariance}
RB schemes are known to be \emph{gauge invariant}: that an implementation map $\phi$ only needs to be close to a conjugation of the ideal map $\omega$ in order to exhibit the desired exponential decay, even if $\phi$ and $\omega$ were far apart under diamond norm. This is still true in the URB setting, and gives us the following alternative definition for the near-ideal implementation.
\begin{dfn}[Near-ideal implementation under gauge]
A URB scheme $R = (S,\mu,\phi,M,\rho_0)$ with $\epsilon$ approximate factoring of $M$ into $(M_0,\phi^*,\mathcal{I})$ is said to have a $(\delta,\kappa)$ near-ideal implementation, if there exists an \emph{ideal map} $\omega:S\rightarrow \mathcal{C}(d)$, and \emph{gauges} $\mathcal{U},\mathcal{V}$ being invertible real superoperators, such that $$||| \mathcal{U}|||_\diamond||| \mathcal{U}^{-1}|||_\diamond||| \mathcal{V}|||_\diamond||| \mathcal{V}^{-1}|||_\diamond\leq \kappa,$$ and 
$$\mathbb{E}_{g\sim \mu}\big[\Vert \mathcal{U}\circ\phi(g)\circ\mathcal{U}^{-1}-\omega(g)\Vert_{\diamond} +\Vert \mathcal{V}\circ\phi^*(g)\circ\mathcal{V}^{-1}-\omega(g)^{\dagger}\Vert_{\diamond} \big]\leq \delta.$$
\end{dfn}

With the definition of near-ideal implementation under gauge, we can further define that a URB scheme $R$ is a $(\epsilon,\delta, \gamma,\kappa)$ twirling scheme under gauge if it has such a near-ideal implementation. The gauge free version corresponds to the case that $\mathcal{U}=\mathcal{V}=\mathrm{id}$ and $\kappa=1$. Our single exponential decay results for different norms can be readily applied to twirling schemes under gauge. We have the following.
\begin{cor}[Single exponential decay under gauge]\label{cor:gauge}
Let $R=(S,\mu,\phi,M,\rho_0)$ be an $(\epsilon,\delta,\gamma,\kappa)$ twirling scheme with respect to the tuple $(M_0,\phi^*,\mathcal{I},\omega,\mathcal{U},\mathcal{V})$, under the diamond norm or the trace norm.
Given that $\delta \leq\frac{1-\gamma}{11}$, there exists $A, B \in \mathbb{R}$ and $p \in [1-2\delta,1]$ such that
\begin{align*}
    \vert p_R(m) - (A + B p^m) \vert \le \epsilon + 16\kappa(\gamma+6\delta)^m.
\end{align*}
Alternatively, let $R=(S,\mu,\phi,M,\rho_0)$ be an $(\epsilon,\delta,\gamma,\kappa)$ twirling scheme on quantum systems with dimension $d$ with respect to the tuple $(M_0,\phi^*,\mathcal{I},\omega,\mathcal{U},\mathcal{V})$ under the Frobenius norm. Let $\delta':=d\cdot \delta$.
Given that $\delta'\leq \frac{1-\gamma}{11}$, there exists $A, B \in \mathbb{R}$ and $p \in [1-2\delta',1]$ such that
\begin{align*}
    \vert p_R(m) - (A + B p^m) \vert \le \epsilon + 16\kappa d^{3/2}(\gamma+6\delta')^m.
\end{align*}
Furthermore, in the case that $\phi^*$ or $\phi$ maps to unitary mixtures we can restrict to $\delta'= \sqrt{d}\cdot\delta$.
\end{cor}
\begin{proof}
    The proofs closely mimic the proofs for \Cref{thm:main}, \Cref{cor:iitrace} and \Cref{cor:spectral}. The central difference is to substitute the use of $\Lambda_R$ with a \emph{gauge-corrected twirling map}
    $$\tilde{\Lambda}_R:=\int_{g\sim \mu}dg\mathcal{V}\circ\phi^*(g)\circ\mathcal{V}^{-1}\circ\cdot\circ \mathcal{U}\circ\phi(g)\circ\mathcal{U}^{-1},$$
    with the observation that
    $$\tr[M_0\cdot \Lambda_R^m(\mathcal{I})(\rho_0)]= \tr[M_0\cdot \mathcal{V}^{-1} \circ\tilde{\Lambda}_R^m(\mathcal{V}\circ\mathcal{I}\circ\mathcal{U}^{-1})\circ\mathcal{U}(\rho_0)].$$
    By treating $\tilde{\Lambda}_R$ as a perturbed version of $\Lambda^*_R$, we can similarly get a diagonalization $(L_1,L_2,R_1,R_2,A'_1,A'_2)$ of $\tilde{\Lambda}_R$. This gives
    \begin{align}
        &\tr[M_0\cdot \Lambda^m_R(\mathcal{I})(\rho_0)] \nonumber\\
        &=\tr[M_0\cdot\mathcal{V}^{-1}\circ L_1(A'_1)^mR_1(\mathcal{V}\circ\mathcal{I}\circ\mathcal{U}^{-1})\circ\mathcal{U}(\rho_0)]\nonumber\\
        &+\tr[M_0\cdot\mathcal{V}^{-1}\circ L_2(A'_2)^mR_2(\mathcal{V}\circ\mathcal{I}\circ\mathcal{U}^{-1})\circ\mathcal{U}(\rho_0)].\label{eqn:two_terms_gauge}
    \end{align}
    The second term is an exponential decay term; we have
    \begin{align*}
        &\tr[M_0\cdot\mathcal{V}^{-1}\circ L_2(A'_2)^mR_2(\mathcal{V}\circ\mathcal{I}\circ\mathcal{U}^{-1})\circ\mathcal{U}(\rho_0)]\\
        \leq &||| \mathcal{V}^{-1}\circ L_2(A'_2)^mR_2(\mathcal{V}\circ\mathcal{I}\circ\mathcal{U}^{-1})\circ\mathcal{U} |||_\diamond\\
        \leq &||| \mathcal{V}^{-1}|||_\diamond\cdot |||\mathcal{U} |||_\diamond\cdot||| L_2(A'_2)^mR_2|||_\diamond\cdot||| \mathcal{V}\circ\mathcal{I}\circ\mathcal{U}^{-1}|||_\diamond\\
        \leq &\kappa  \cdot||| L_2(A'_2)^mR_2|||_\diamond,
    \end{align*}
    from which \Cref{eqn:errnorm} or \Cref{eqn:errnorm_spec} apply. 
    
    The analysis of the first term is similar to that of the proof of \Cref{thm:main}.
\end{proof}
\subsubsection{Relating the Decay Rate to the Average Fidelity}
The decay rate in an RB or URB scheme is often believed to indicate the average ``quality'' of a collection of gate implementations. Mathematically, it corresponds to the second largest eigenvalue (guaranteed to be real when the URB scheme is a twirling scheme) of the twirl operator $\Lambda_R$. Relating such an eigenvalue with an indicator of the average ``quality'' of gates can be non-obvious and sometimes tricky~\cite{proctor2017randomized}. Of course, we can always avoid this question by defining a URB scheme's decay rate as a figure of merit by fiat, but we can consider some cases where there is a relatively clear connection between the average fidelity and the second largest eigenvalue.

During the proof of~\Cref{thm:main}, we see that the eigen-channel with eigenvalue 1 is always a replacement channel $\mathcal{E}_{\rho^*}$. All other eigen-superoperators, including the one corresponding to the second largest eigenvalue, lie in the space of channel differences $V_0(d)$. To see this, note that any eigne-channel $\mathcal{N}\in\mathcal{C}(d)$ of $\Lambda_R$ must be of eigenvalue $1$ (otherwise $\Lambda_R(\mathcal{N})$ is no longer a channel) and consequently corresponds to an eigen-superoperator $\mathcal{N}-\mathcal{E}_{\rho^*}\in V_0(d)$. Let $\mathcal{D}_R\in V_0(d)$ be the eigen-superoperator corresponding to the eigenvalue $p_R=\lambda_2$.

An easy case that $p$ has a natural interpretation is when the line $\{\mathcal{E}_{\rho^*}+\lambda\cdot \mathcal{D}_R|\lambda\in\mathbb{R}\}$ contains a unitary channel $\mathcal{U}$. In this case we can perform a unitary gauge transformation $\phi(g)\mapsto \mathcal{U}\circ\phi(g)\circ\mathcal{U}^{-1}$ so that we can assume without loss of generality that the unitary channel is the identity channel $\mathrm{id}$.\footnote{This can be done in general when the line $\{\mathcal{E}_{\rho^*}+\lambda\cdot \mathcal{D}_R|\lambda\in\mathbb{R}\}$ contains an invertible superoperator $\mathcal{R}$. However in this case the gauge transformation $\phi(g)\mapsto \mathcal{R}\circ\phi(g)\circ\mathcal{R}^{-1}$ may no longer preserve channels, making the physical interpretation of the decay rate less physical.} Then
$$\Lambda_R(\mathrm{id}) = p\cdot \mathrm{id} + (1-p)\cdot \mathcal{E}_{\rho^*}.$$
The decay rate $p$ is then related to the fidelity of the channel $\Lambda_R(\mathrm{id}) = \int_{g\sim u}dg \phi^*(g)\phi(g)$, which in turn is the average fidelity of the channels $\phi^*(g)\phi(g)$ under the probability distribution $\mu$, specifically,
$$p_R = \frac{d\mathbb{E}_{g\sim \mu} [F_{avg}(\phi^*(g)\phi(g))]-1}{d-1},$$
where the average fidelity of a channel $\mathcal{C}$ is defined as $$F_{avg}(\mathcal{C}):=\int_{U}dUF
\big(U|0\rangle\langle 0|U^\dagger, \mathcal{C}(U|0\rangle\langle 0|U^\dagger)\big)$$ over Haar random unitaries $U$. One example of an error model satisfying this condition is the \emph{gate-dependent replacement model}: that there exists probabilities $p(g)$ and states $\rho(g)$ for each $g$, such that
$$\phi^*(g)\phi(g)=p(g)\cdot \mathrm{id} + (1-p(g))\cdot \mathcal{E}_{\rho(g)},$$ $$\int_{g\in\mu}dg(1-p(g))\rho(g)\propto \rho^*.$$

For gate independent noise models $$\phi^*(g) = \mathcal{N}_{L}^{(1)}\circ \omega(g)^\dagger\circ \mathcal{N}_L^{(2)}, \phi(g) = \mathcal{N}_{R}^{(1)}\circ \omega(g)\circ \mathcal{N}_R^{(2)},$$ we can without loss of generality consider either \emph{in-between noise models}
$$\phi^*(g) = \omega(g)^\dagger\circ \mathcal{N}_L, \phi(g) = \mathcal{N}_{R}\circ \omega(g)$$
or \emph{sandwiched noise models}
$$\phi^*(g) = \mathcal{N}_L\circ\omega(g)^\dagger, \phi(g) = \omega(g)\circ\mathcal{N}_{R}$$
through gauge transformation, where $\mathcal{N}_L:=\mathcal{N}_L^{(2)}\circ \mathcal{N}_{L}^{(1)}$ and $\mathcal{N}_{R}:=\mathcal{N}_R^{(2)}\circ \mathcal{N}_{R}^{(1)}$.
\begin{itemize}
    \item For the \emph{in-between noise model}, we have $$\Lambda_R(\mathrm{id}) = \Lambda^*_R(\mathcal{N}_L\circ \mathcal{N}_R).$$
    In the case that $\mu$ gives rise to a unitary 2-design, or equivalently $\Lambda_R^*=\Lambda^*$, we have
    $$\Lambda_R(\mathrm{id}) = \Lambda^*_R(\mathcal{N}_L\circ \mathcal{N}_R)=p\cdot \mathrm{id} + (1-p)\cdot\mathrm{dep},$$
    where $p_R = \frac{d\cdot F_{avg}(\mathcal{N}_L\circ\mathcal{N}_R)-1}{d-1}$. This recovers the case that RB on unitary two-designs extract the average fidelity of gate-independent noises.
    \item For the \emph{sandwiched noise model}, we have $$\Lambda_R(\mathrm{id}) = \mathcal{N}_L\circ\mathcal{N}_R.$$
    Then the RB exponent $p_R$ relates to the average fidelity of $\mathcal{N}_L\circ\mathcal{N}_R$ if the following holds:
    $$\exists p\in[0,1], \mathcal{N}_L\circ\mathcal{N}_R=p\cdot\mathrm{id} + (1-p)\cdot \mathcal{E}_{\rho^*},$$
    where $\rho^*$ is the eigenstate of the channel $\int_{g\sim \mu}dg\phi^*(g)=\mathcal{N}_L\circ\int_{g\sim\mu}dg\omega(g)^\dagger.$ This also guarantees that $p_R=p$.
\end{itemize}

To summarize, for gate-independent noise models, the RB exponent extracted from the experiment is not always determined by the average fidelity of the error channels, partially because that the ideal twirl is not always a full twirl. Special cases where there is such a determination is when the ideal twirl is a full twirl, and when the noise channels adds up to an replacement channel with the maximal eigenstate.

\subsection{Analysis on Robustness of Data Fitting}
\label{subsec:robustness}
In this section, we will prove that a sufficiently small perturbation to an exponential decay curve will not significantly affect the extracted decay rate. More specifically, we will prove the following lemma:
\begin{lem}
\label{lem:robustness}
Let $C_0(m)=A_0 \alpha^m +B_0$ be an exponential decay curve. Given $\epsilon \ll A_0\alpha^{M}/2$, if there is another exponential curve $C_1(m)=A_1 \beta^m+B_1$ such that $\abs{C_1(m)-C_0(m)} \leq \epsilon$ for all integers $m\geq M$, then the decay rates $\beta$ and $\alpha$ must satisfy $$\abs{\beta-\alpha}\leq 2 \epsilon\frac{\alpha+1}{A_0 \alpha^{M}-2\epsilon}.$$
\end{lem}
\begin{proof}
Since $|B_0-B_1|=\lim_{m \rightarrow\infty}|C_0(m)-C_1(m)|\leq \epsilon$, we have $\abs{A_1 \beta^m-A_0 \alpha^m}\leq 2\epsilon$ for any integer $m\geq M$. 

Assume without loss of generality that $1>\beta\geq \alpha>0$. Then we have
\begin{equation}
\begin{split}
&    \beta (A_0 \alpha^{M}-2\epsilon)\leq \beta A_1\beta^{M}=A_1 \beta^{M+1} \leq A_0 \alpha^{M+1}+2\epsilon\\
\rightarrow&  \beta-\alpha \leq \frac{A_0 \alpha^{M+1}+2\epsilon}{A_0 \alpha^{M}-2\epsilon}-\alpha=\frac{2\epsilon \alpha+2\epsilon}{A_0 \alpha^{M}-2\epsilon}.
\end{split}
\end{equation}
\end{proof}

\Cref{lem:robustness} bounds the deviation of the decay rate, given that the fitted curve is $O(\epsilon/A_0)$-close to a certain ground truth single exponential decay curve under the $\ell_{\infty}$-norm. In our case, such $\ell_\infty$ deviations come from the error terms $\epsilon + 16(\gamma+6\delta)^m$, standard deviation from sampling $\Theta\left(\frac{\log k}{\sqrt{K}}\right)$ ($k$ being the number of sequence lengths chosen and $K$ the number of repeats for each sequence length), and imperfect fitting algorithms. Even though finite $\epsilon$ makes the error term non-vanishing, we conclude that unless the number of samples for each sequence length exceeds $\epsilon^{-2}$, the inaccuracy of decay rate extraction mainly comes from stochastic fluctuations of the random experiments.

\section{Discussion and Open Questions}
\label{sec:discussion}
In this paper we propose a new framework for randomized benchmarking, which we call the URB framework. To formulate this we take major components of widely used benchmarking schemes such as group-based RB and linear XEB and then generalize them. With this generalization we can express a much wider variety of schemes that could possibly overcome the shortcomings of current schemes, such as scalability. For a certain class of URB schemes, which we call twirling schemes, we can prove an exponential decay using~\Cref{thm:main}, a hallmark feature of randomized benchmarking schemes. We saw that the crucial property we need to establish to apply~\Cref{thm:main} is for the twirling operator $\Lambda_R^*$ to be a $\gamma$-approximate twirl for $\gamma<1$. We gave upper bounds for $\gamma$ and give alternative exponential decay results for easier bounds to compute, such as the $|||\cdot |||_2$ norm. We also discuss issues of gauge invariance and the interpretation of the decay rate. 

One recent RB variant, called randomized benchmarking with mirror circuits~\cite{proctor2021scalable}, aims to extract system fidelity information from shallow Clifford circuits. The random gate sequence consists of 
\begin{enumerate}
    \item Clifford layers with a mirror structure, i.e. the second half of the Clifford circuits are inversions of the first half
    \item Uniform Pauli gates interleaving the Clifford layers
    \item Single-qubit Cliffords at the beginning and the end of the circuits.
\end{enumerate}
 This falls into our RB framework as all gates after the first half of Clifford layers can be seen as the post-processing POVM. However, with the presence of the interleaving Pauli gates, it is impossible to separate out ``inverting maps'' that  locally approximately inverts the implementation maps. Hence, this scheme does not readily yield to our analysis. In this scheme, the Pauli error is corrected on a global scale after the final measurement. We leave it to future work to analyze such schemes involving complex correlations between the random gates.

We leave some open questions for future work. Obviously, the value of our framework is really proved by novel benchmarking schemes that have better performance by bypassing the group-based requirement. One work towards this direction is~\cite{chen2022linear}, which proposes performing linear XEB with shallow Clifford circuits so that the classical simulation step is much more scalable to larger numbers of qubits. They simulate systems with more than a thousand qubits to showcase this scalability. An earlier work~\cite{boone2019randomized} also considered a variant of RB that does not form a group.

The upper bounds on $\gamma$ are somewhat loose and could be analyzed more systematically. In particular, it would be interesting to express the induced diamond norm as an SDP so that it can be bounded or even exactly calculated. Furthermore, various inequalities of the proof of~\Cref{thm:main} could be tightened, such as improving the factor of 11 in the $\delta \leq (1-\gamma)/11$ bound and a possible typicality argument for the bounding of $\Vert (A_2')^m\Vert$. Also, our result on robustness of data fitting naturally leads to a question of sample complexity, which we also leave for future work. As discussed above, the interpretation of the decay rate is still a matter of discussion. Also, the exponential decay of the $F(m)$ quantity for linear XEB defined in~\Cref{eq:linear_xeb} is still to be further investigated. Lastly, from our work we recognize the value of a systematic formulation of norms for multilinear algebra as applied to quantum information science.

Another direction to pursue is to recover matrix exponential decay results for group-based RB~\cite{wallman2018randomized, helsen2020general, merkel2021randomized, kong2021framework}. There are two major differences between our result and Fourier-analysis based results for group-based RB. First, we treat the final recovery gate as an imperfect implementation of the perfect recovery gate, which enables factorization and hence the study of twirling maps, but this introduces a constant error term which is not present in the Fourier-based analyses. Intuitively, this is because our twirling map approach fails to capture the correlation between the noise in the final recovery gate with those of the previously applied gates due to taking the expectation over the group. This removes the constant error term in Fourier analysis-based approaches. We leave the incorporation of such correlated noise in a twirling maps approach to future work. Also, the case where multiple irreducible representations are concerned corresponds to the case where $\Lambda^*_R\neq \Lambda^*$ is a projector on a larger space and is thus not an approximate twirl. Recovering the matrix exponential result then requires an investigation into $\Lambda^*_R$ that do not form approximate twirls. However, our framework is not concerned with group structure or irreducible representations, but only the twirling map. This could provide a conceptually easier approach when dealing with infinite groups.

\section*{Acknowledgements}
We would like to thank Chunqing Deng, Linghang Kong, Yaoyun Shi, and Tenghui Wang for helpful comments. DD would like thank God for all of His provisions.

\begin{appendices}
\section{Properties of the $\alpha$ Matrix and $\beta$ Tensor}
\label{sec:app_alpha_beta}
We first recall that we have a basis $\{P_i\}_{i=1}^{d^2}$ of $L(\mathbb{C}^d)$ that satisfies
\begin{align*}
    \tr[P_i P_j] = d \delta_{ij}.
\end{align*}
We can use this basis to express superoperators 
\begin{align*}
    \mathcal T = \sum_{i,j} \alpha_{i,j} P_i \cdot P_j
\end{align*}
and supersuperoperators
\begin{align*}
    \Lambda: P_i \cdot P_j \mapsto \sum_{k,l} \beta_{i,j}^{k,l} P_k \cdot P_l.
\end{align*}
We establish some facts relating the norms of the $\alpha$ matrix with the norms of its corresponding superoperator. We will use these results when proving the bounds on $\gamma$.
\begin{prp}
\label{prp:alpha_hs_t_so}
Let $\mathcal T$ be a superoperator. Then,
\begin{align*}
    \Vert \mathcal T\Vert_\mathrm{SO} = d \Vert \alpha(\mathcal T) \Vert_\mathrm{HS},
\end{align*}
where $\Vert \cdot \Vert_\mathrm{SO}$ is the norm on superoperators induced by the $\langle \cdot, \cdot\rangle_\mathrm{SO}$ inner product.
\end{prp}
\begin{proof}
    Since the Pauli operators $\{P_m\}_{m=1}^{d^2}$ are orthogonal under the Hilbert Schmidt inner product, we can directly compute
    \begin{align*}
        \Vert \mathcal T \Vert_\mathrm{SO}^2 & = \frac{1}{d}\sum_{m=1}^{d^2} \langle \mathcal T(P_m),\mathcal T(P_m)\rangle_\mathrm{HS} \\
        & = \frac{1}{d} \sum_{m=1}^{d^2} \tr\left[\sum_{i,j,k,l=1}^{d^2} \alpha_{i,j}^* \alpha_{k,l} P_j P_m P_i P_k P_m P_l\right] \\
        & = \frac{1}{d} \tr\left[\sum_{i,j,k,l=1}^{d^2} \alpha_{i,j}^* \alpha_{k,l} P_j \sum_{m=1}^{d^2} (P_m P_i P_k P_m )P_l\right]\\
        & = \frac{1}{d} \tr\left[\sum_{i,j,k,l=1}^{d^2} \alpha_{i,j}^* \alpha_{k,l} P_j (d \tr[P_i P_k] I)P_l\right]\\
        & =  d \tr\left[\sum_{i,j,k,l=1}^{d^2} \alpha_{i,j}^* \alpha_{k,l} \delta_{ik} P_j P_l\right]\\
        & =  d^2 \sum_{i,j,k,l=1}^{d^2} \alpha_{i,j}^* \alpha_{k,l} \delta_{ik} \delta_{jl} \\
        & = d^2 \sum_{i,j=1}^{d^2} \vert \alpha_{i,j} \vert^2.
    \end{align*}
    In the fourth equality we used for any operator $X$,
    \begin{align*}
        \sum_{m=1}^{d^2} P_m X P_m & = d \tr[X] I.
    \end{align*}
    This directly gives our conclusion.
\end{proof}
\begin{prp}
\label{prp:alpha_channel}
Let $\mathcal N$ be a quantum channel. Its $\alpha$ matrix is positive semidefinite with unit trace.
\end{prp}
\begin{proof}
    Let $\mathcal N$ be a quantum channel. Then, it has a Kraus representation
    \begin{align*}
        \mathcal N = \sum_{k=1}^K A_k \cdot A_k^\dagger,
    \end{align*}
    where $K \leq d^2$ is the Kraus rank and 
    \begin{align*}
        \sum_{k=1}^K A_k^\dagger A_k = I.
    \end{align*}
    Suppose the decomposition of $A_k \in L(\mathbb{C}^d)$ in terms of $\{P_i\}_i$ is given by
    \begin{align*}
        A_k = \sum_i a_{k,i} P_i.
    \end{align*}
    Then,
    \begin{align*}
        \alpha(\mathcal N)_{i,j} = \sum_{k=1}^K a_{k,i} a_{k,j}^*.
    \end{align*}
    This implies $\alpha$ is a Gram matrix, which is equivalent to being positive semidefinite. Furthermore,
    \begin{align*}
        \sum_{k=1}^K A_k^\dagger A_k = \sum_{i,j} \alpha_{j,i} P_i P_j = I.
    \end{align*}
    Taking traces on both sides,
    \begin{align*}
        \sum_{i,j} \alpha_{j,i} d\delta_{ij} = d.
    \end{align*}
    Hence, $\alpha$ has unit trace.
\end{proof}
\noindent This is a nice property as intuitively one would expect a quantum channel to be a higher-order tensor analogue of a quantum state. However, we can't take this analogy that far since the converse is not true. For example, suppose $d=2$ and we choose the order $\{I,X,Y,Z\}$ with 
\begin{align*}
    \alpha = 
    \begin{bmatrix}
    \frac 1 2 & \epsilon & 0 & 0\\
    \epsilon & \frac 1 2 & 0 & 0\\
    0 & 0 & 0 & 0\\
    0 & 0 & 0 & 0
    \end{bmatrix}.
\end{align*}
For small positive $\epsilon$, $\alpha$ is positive semidefinite with unit trace. However, 
\begin{align*}
    \sum_{i,j} \alpha_{j,i} B_i B_j = I +  2\epsilon X \neq I.
\end{align*}

We next consider the $\beta$ tensor for twirling maps corresponding to a measure. We write out:
$$\Lambda(\mu) : P_i \cdot P_j \mapsto \int_{g \sim \mu} dg U_g^\dagger P_i U_g \cdot U_g^\dagger P_j U_g.$$
Denote
$$U_g^\dagger P_i U_g = \sum_{k=1}^{d^2} u(g)_{i,k} P_k.$$
Since the Paulis are Hermitian and orthogonal, $u(g)_{i,k} \in \mathbb{R}$ and $\sum_{k=1}^{d^2} \vert u(g)_{i,k}\vert^2 =1$. Then,
$$\Lambda(\mu) : P_i \cdot P_j \mapsto \int_{g \sim \mu} dg U_g^\dagger P_i U_g \cdot U_g^\dagger P_j U_g = \int_{g \sim \mu} dg \sum_{k,l=1}^{d^2} u(g)_{i,k} u(g)_{j,l} P_k \cdot P_l,$$
so 
\begin{align*}
    \beta_{i,j}^{k,l}(\mu) = \int_{g \sim \mu} dg u(g)_{i,k} u(g)_{j,l} = \frac{1}{d^2}\int_{g \sim \mu} dg \, tr[U_g^\dagger P_i U_g P_k]tr[U_g^\dagger P_j U_g P_l].
\end{align*}
This proves $\beta_{i,j}^{k,l}(\mu) \in \mathbb{R}$. From this representation we also obtain the identity:
$$\beta_{i,j}^{k,l}(\mu) = \beta_{j,i}^{l,k}(\mu)$$
as well as the fact
\begin{align}
\label{eq:beta_II}
    \Lambda(\mu)(I \cdot I) = I \cdot I.
\end{align}
Furthermore, since $\Lambda(\mu)$ preserves the set of channels, if $\mathcal N$ is a channel,
$$\Lambda(\mu)(\mathcal N) = \sum_{i,j=1}^{d^2} \alpha_{i,j} \sum_{k,l = 1}^{d^2} \beta_{i,j}^{k,l} P_k \cdot P_l = \sum_{k,l= 1}^{d^2} (\sum_{i,j=1}^{d^2} \alpha_{i,j} \beta_{i,j}^{k,l}(\mu)) P_k \cdot P_l$$
is also a channel, which puts additional constraints on $\beta(\mu)$. 

Lastly, we compute the $\beta$ tensor for the Haar measure: $\beta(\eta)$. We can use the expression for a Haar twirl from~\cite{emerson2005scalable} to compute
\begin{align*}
\Lambda(\mu)(P_i \cdot P_j)(\rho)& = \frac{\tr[P_iP_j] \tr[\rho] }{d} \frac{I}{d} +\frac{d\tr[P_i]\tr[P_j] - \tr[P_i P_j]}{d (d^2-1)} \left(\rho - \tr[\rho]\frac{I}{d}\right).
\end{align*}
If exactly one of $P_i,P_j$ is $I$, then the RHS is zero. If they're both not $I$, 
\begin{align*}
\Lambda(\mu)(P_i \cdot P_j)(\rho)& = \delta_{ij} \tr[\rho]\frac{I}{d}  -\frac{\delta_{ij}}{d^2-1} \left(\rho- \tr[\rho] \frac{I}{d}\right)\\
& = \frac{\delta_{ij}}{d^2-1} \left[ -\rho +d^2 \tr[\rho] \frac{I}{d}\right] \\
& = \frac{\delta_{ij}}{d^2-1} \left[ - I\rho I + \sum_{k=1}^{d^2} P_k \rho P_k\right] \\
& = \frac{\delta_{ij}}{d^2-1} \sum_{k=2}^{d^2} P_k \rho P_k.
\end{align*}
Along with~\Cref{eq:beta_II}, we conclude
\begin{align}
\label{eq:beta_haar}
\beta_{i,j}^{k,l}(\eta) = 
\begin{cases}
1 & i=j=k=l=1\\
\frac{1}{d^2-1} & (i=j > 1 )\land( k=l>1) \\
0 & \text{otherwise}
\end{cases}.
\end{align}
In words, $\Lambda(\eta)$ maps a pair of the same non-identity Pauli's to a uniform convex combination of same non-identity Pauli's.

\section{$\ell_2$-Norm Like Bound}
\label{sec:app_l2_bound}
Here we prove~\Cref{prp:l2_norm_bound}:
\begin{proof}
    Let $\mathcal T$ be a superoperator. We compute
    \begin{align*}
    & \Vert [\Lambda(\mu) - \Lambda(\mu_H) ](\mathcal T)\Vert_\diamond\\
    & = \max_{\Vert \rho \Vert_1 =1} \Vert (\Lambda(\mu)(\mathcal T) \otimes \mathbb{I} - \Lambda(\mu_H)(\mathcal T) \otimes \mathbb{I})(\rho)  \Vert_1\\
    & = \max_{\Vert \rho \Vert_1=1} \left\Vert \sum_{\substack{k,l=1}}^{d^2} \left[\sum_{i,j=1}^{d^2} \alpha_{i,j} \left(\beta_{i,j}^{k,l}(\mu)-\beta_{i,j}^{k,l}(\mu_H)\right) \right] (P_k \otimes I)\rho (P_l\otimes I) \right\Vert_1 \\
    & \le \sum_{\substack{k,l=1}}^{d^2} \left\vert\sum_{i,j=1}^{d^2} \alpha_{i,j} \left(\beta_{i,j}^{k,l}(\mu)-\beta_{i,j}^{k,l}(\mu_H)\right) \right\vert\max_{\Vert \rho \Vert_1=1} \left\Vert (P_k \otimes I)\rho (P_l\otimes I) \right\Vert_1 \\
    & \le \Vert \alpha\Vert_{\ell_2}   \sqrt{\sum_{\substack{k,l,k',l'=1}}^{d^2} \left\vert \sum_{i,j=1}^{d^2} (\beta^{k,l}_{i,j}(\mu)-\beta^{k,l}_{i,j}(\mu_H))(\beta^{k',l'}_{i,j}(\mu)-\beta^{k',l'}_{i,j}(\mu_H))\right\vert} \\
    & = \frac{1}{d} \Vert \mathcal T \Vert_\mathrm{SO} \sqrt{\sum_{\substack{k,l,k',l'=1}}^{d^2} \left\vert \sum_{i,j=1}^{d^2} (\beta^{k,l}_{i,j}(\mu)-\beta^{k,l}_{i,j}(\mu_H))(\beta^{k',l'}_{i,j}(\mu)-\beta^{k',l'}_{i,j}(\mu_H))\right\vert} \\
    &\leq  2 \Vert \mathcal T \Vert_\diamond\sqrt{\sum_{\substack{k,l,k',l'=1}}^{d^2} \left\vert \sum_{i,j=1}^{d^2} (\beta^{k,l}_{i,j}(\mu)-\beta^{k,l}_{i,j}(\mu_H))(\beta^{k',l'}_{i,j}(\mu)-\beta^{k',l'}_{i,j}(\mu_H))\right\vert},
    \end{align*}
    where the second inequality involves the following argument:
    \begin{align*}
        \sum_m |\langle a, b_m\rangle| &=  \max_{k_m \in U(1)}\langle a, \sum_m  k_m\cdot b_m \rangle \\
        &\leq \Vert a\Vert_{\ell_2}\cdot \max_{k_m\in U(1)}\Vert \sum_m  k_m\cdot b_m\Vert_{\ell_2}\\
        & = \Vert a\Vert_{\ell_2}\cdot \max_{k_m, k_{m'} \in U(1)}\sqrt{\sum_{m,m'}(  k_m\cdot k_{m'}) \langle b_m,b_{m'}}\rangle\\
        &\leq \Vert a\Vert_{\ell_2}\cdot\sqrt{\sum_{m,m'}| \langle b_m, b_{m'}\rangle|},
    \end{align*}
    the third equality follows from~\Cref{prp:alpha_hs_t_so}, and the last inequality follows from~\Cref{eq:SO_tr_inequality}.
\end{proof}

\section{Induced $\ell_1$-Norm Like Bound}
\label{sec:app_induced_l1_bound}
Here we prove~\Cref{prp:induced_l1_bound}:
\begin{proof}
    We first establish that 
    \begin{align*}
        \Vert \Lambda(\mu)(\mathcal T) \Vert_\diamond = \left\Vert \int_{g \sim \mu} \omega(g^{-1}) \circ \mathcal T \circ \omega(g) dg \right \Vert_\diamond \leq \int_{g \sim \mu} dg \Vert  \omega(g^{-1})\Vert_\diamond \Vert \mathcal T \Vert_\diamond \Vert \omega(g) \Vert_\diamond = \Vert \mathcal T \Vert_\diamond  .
    \end{align*}
    Furthermore, we have that
    \begin{align}
        \Lambda(\mu_P)(\mathcal T) & = \Lambda(\mu_P) \left(\sum_{i,j} \alpha_{i,j} P_i \cdot P_j\right) \nonumber\\
        & = \frac{1}{d^2} \sum_k \sum_{i,j} \alpha_{i,j} P_k P_i P_k \cdot P_k P_j P_k \nonumber\\
        & = \frac{1}{d^2} \sum_k \sum_{i,j} \alpha_{i,j} (-1)^{\langle b_k, b_i \rangle_S + \langle b_k, b_j \rangle_S} P_i \cdot P_j \nonumber\\
        & = \frac{1}{d^2} \sum_k \sum_{i,j} \alpha_{i,j} (-1)^{\langle b_k, b_i \oplus b_j \rangle_S} P_i \cdot P_j \nonumber\\
        & = \sum_i \alpha_{i,i}(\mathcal T) P_i \cdot P_i \label{eq:pauli_twirl},
    \end{align}
    where $\langle \cdot ,\cdot\rangle_S$ denotes the symplectic inner product of the binary symplectic representations of a Pauli string~\cite{gottesman1997stabilizer} and $\oplus$ denotes bitwise addition. Two Paulis commute iff the inner product vanishes. The fourth equality follows from the distributive property of the inner product. The last inequality follows because when $i=j$, $b_i\oplus b_j=0$, while when $i\neq j$, half of the summed $b_k$ is orthogonal to $b_i \oplus b_j$ while the other half is not and therefore cancels out.
    
    Since $\mathcal N_i$ is a channel, $\alpha(\mathcal N_i)$ has unit trace. Thus, by~\Cref{eq:pauli_twirl}, $\Lambda(\mu_P)(\mathcal N_i)$ is a Pauli channel, that is, a mixture of Pauli unitaries, and $\Lambda(\mu_P)(\mathcal T)$ is a difference of Pauli channels. Hence we conclude
    \begin{align}
        \Vert \mathcal T \Vert_\diamond & \geq \Vert \Lambda(\mu_P)(\mathcal T) \Vert_\diamond \nonumber\\
        & = \Vert \mathrm{diag}(\alpha(\mathcal T) )\Vert_{\ell_1} \label{eq:data_processing},
    \end{align}
    where the equality is proven in~\cite{sacchi2005optimal}.

    By~\Cref{eq:beta_haar},
    \begin{align*}
        \Lambda(\eta) = \Lambda(\eta) \circ \Lambda(\mu_P).
    \end{align*}
    We can therefore argue
    \begin{align*}
    & \Vert [\Lambda(\mu) - \Lambda(\eta) ](\mathcal T)\Vert_\diamond\\
    & =\Vert [\Lambda(\mu) - \Lambda(\eta) ] \circ \Lambda(\mu_P)(\mathcal T)\Vert_\diamond\\
    & = \max_{\Vert \rho \Vert_1 =1} \Vert [\Lambda(\mu)( \Lambda(\mu_P)(\mathcal T) )\otimes \mathbb{I} - \Lambda(\eta)( \Lambda(\mu_P)(\mathcal T)) \otimes \mathbb{I}](\rho)  \Vert_1\\
    & = \max_{\Vert \rho \Vert_1=1} \left\Vert \sum_{\substack{k,l=1}}^{d^2} \left[\sum_{i=1}^{d^2} \alpha_{i,i}(\mathcal T) \left(\beta_{i,i}^{k,l}(\mu)-\beta_{i,i}^{k,l}(\eta)\right) \right] (P_k \otimes I)\rho (P_l\otimes I) \right\Vert_1 \\
    & \le \sum_{\substack{k,l=1}}^{d^2} \left\vert\sum_{i=1}^{d^2} \alpha_{i,i}(\mathcal T) \left(\beta_{i,i}^{k,l}(\mu)-\beta_{i,i}^{k,l}(\eta)\right) \right\vert\max_{\Vert \rho \Vert_1=1} \left\Vert (P_k \otimes I)\rho (P_l\otimes I) \right\Vert_1 \\
    & \le  \sum_{i=1}^{d^2} \left\vert\alpha_{i,i}(\mathcal T) \right\vert\sum_{\substack{k,l=1}}^{d^2}\vert\beta_{i,i}^{k,l}(\mu)-\beta_{i,i}^{k,l}(\eta)\vert  \\
    & \le \max_{i \in [d^2]} \sum_{\substack{k,l=1}}^{d^2}  \left\vert\beta_{i,i}^{k,l}(\mu)-\beta_{i,i}^{k,l}(\eta) \right\vert  \Vert\mathrm{diag}(\alpha(\mathcal T))\Vert_{\ell_1}  \\
    & \leq  \max_{i \in [d^2]} \sum_{\substack{k,l=1}}^{d^2}  \left\vert\beta_{i,i}^{k,l}(\mu)-\beta_{i,i}^{k,l}(\eta) \right\vert \Vert \mathcal T \Vert_\diamond,
    \end{align*}
    where the last inequality uses~\Cref{eq:data_processing}.
\end{proof}

\end{appendices}

\bibliographystyle{unsrt}

\end{document}